\documentclass[a4paper, 11pt]{article}
\usepackage{amsmath}
\usepackage{amssymb}
\usepackage{amsthm}
\usepackage{graphicx}
\usepackage{subfigure}
\usepackage{cite}

\newtheorem{prop}{Proposition}
\newcommand{\bm}{\boldsymbol}
\newcommand{\vect}[1]{\mathbf{#1}}
\newcommand{\mat}[1]{\mathbf{#1}}

\newcommand{\E}[1]{{\mathcal{E}}\left\{#1\right\}}

\newcommand{\ex}[1]{\exp \left\{#1\right\}}

\newcommand{\Tr}[1]{{\mathrm{Tr}}\{#1\}}

\newcommand{\Real}[1]{{\mathrm{Re}}\{#1\}}

\newcommand{\Prob}[1]{{\mathrm{Pr}}\left[#1\right]}

\renewcommand{\det}[1]{|#1|}

\newcommand{\tp}{t_{p}}
\newcommand{\ts}{t_{s}}
\newcommand{\Tp}{T_{p}}
\newcommand{\Ts}{T_{s}}
\newcommand{\Qt}{Q_{t}}
\newcommand{\Q}{Q}
\newcommand{\Qtp}{Q_{\tp}}
\newcommand{\Qts}{Q_{\ts}}

\newcommand{\phitrue}{\phi_{0}}

\newcommand{\va}{\vect{a}}
\newcommand{\atrue}{\va_{0}}

\newcommand{\m}{\vect{m}}
\newcommand{\n}{\vect{n}}
\newcommand{\mt}{\m_{\tp}}

\newcommand{\npt}{\vect{n}_{\tp}}
\newcommand{\nst}{\vect{n}_{\ts}}
\newcommand{\vu}{\vect{u}}

\newcommand{\upt}{\vu_{\tp}}
\newcommand{\vpt}{\vect{v}_{\tp}}
\newcommand{\ust}{\vu_{\ts}}
\newcommand{\x}{\vect{x}}

\newcommand{\xp}{\vect{x}}
\newcommand{\xs}{\vect{y}}
\newcommand{\xpt}{\vect{x}_{\tp}}
\newcommand{\xst}{\vect{y}_{\ts}}
\newcommand{\z}{\vect{z}}
\newcommand{\zpt}{\z_{\tp}}
\newcommand{\zst}{\z_{\ts}}
\newcommand{\wpt}{\vect{w}_{\tp}}

\newcommand{\A}{\mat{A}}

\newcommand{\F}{\mat{F}}
\newcommand{\Fp}{\mat{F}_{x}}
\newcommand{\Fs}{\mat{F}_{y}}
\newcommand{\I}{\mat{I}}

\newcommand{\Mp}{\mat{M}}

\newcommand{\R}{\mat{R}}
\newcommand{\Rhat}{\hat{\R}}
\newcommand{\Rml}{\R_{\tiny \text{ML}}}
\newcommand{\Raml}{\R_{\tiny \text{AML}}}
\newcommand{\Rracg}{\R_{\tiny \text{RACG}}}
\newcommand{\iR}{\R^{-1}}

\newcommand{\isqrtR}{\R^{-1/2}}
\newcommand{\Rj}{\R_{j}}
\newcommand{\Rk}{\R_{k}}
\newcommand{\Rnj}{\R^n_{j}}
\newcommand{\Rnk}{\R^n_{k}}

\newcommand{\Porthaw}{\mat{P}^{\perp}_{\isqrtR\va(\phi)}}
\newcommand{\Porthawtrue}{\mat{P}^{\perp}_{\isqrtR\atrue}}

\newcommand{\Xp}{\mat{X}}
\newcommand{\Xs}{\mat{Y}}

\newcommand{\vphi}{\bm{\phi}}
\newcommand{\vtheta}{\bm{\theta}}
\newcommand{\vthetas}{\vtheta^{s}}
\newcommand{\vthetan}{\vtheta^{n}}

\newcommand{\vzeta}{\bm{\zeta}}

\newcommand{\mGamma}{\bm{\Gamma}}

\newcommand{\vCN}[2]{\mathbb{C}\mathcal{N}\left(#1,#2\right)}
\newcommand{\mCN}[3]{\mathbb{C}\mathcal{N}\left(#1,#2,#3\right)}

\newcommand{\Cchisquare}[1]{\mathbb{C}\chi^{2}_{#1}}

\newcommand{\pdfGamma}[2]{\mathcal{G}\left({#1,#2}\right)}

\newcommand{\dmutdjH}{\frac{\partial \mt^{H}}{\partial \theta_j}}

\newcommand{\dmutdsjH}{\frac{\partial \mt^{H}}{\partial \theta^s_j}}

\newcommand{\dmutdk}{\frac{\partial \mt}{\partial \theta_k}}

\newcommand{\dmutdsk}{\frac{\partial \mt}{\partial \theta^s_k}}

\newcommand{\thetaj}{\theta_{j}}
\newcommand{\thetak}{\theta_{k}}

\newcommand{\thetanj}{\theta^{n}_{j}}

\newcommand{\dist}{\overset{d}{=}}

\begin{document}
\title{Bounds for maximum likelihood regular and non-regular DoA estimation in $K$-distributed noise}
\author{Yuri I. Abramovich\thanks{Y. Abramovich is with W R Systems, Ltd., 11351 Random Hills Road, Suite 400, Fairfax, VA 22030, USA. e-mail: yabramovich@wrsystems.com}, Olivier Besson\thanks{O. Besson is with the University of Toulouse, ISAE-Supaero, Department Electronics Optronics Signal, 10 Avenue Edouard Belin, 31055 Toulouse, France. e-mail: olivier.besson@isae-supaero.fr} and Ben A. Johnson\thanks{B. Johnson is with the University of South Australia - ITR, Mawson Lakes SA 5085, Australia. e-mail: ben.a.johnson@ieee.org}}
\date{}
\maketitle

\begin{abstract}
We consider the problem of estimating the direction of arrival of a signal embedded in $K$-distributed noise, when secondary data which contains noise only are assumed to be available. Based upon a recent formula of the Fisher information matrix (FIM) for complex elliptically distributed data, we provide a simple expression of the FIM with the two data sets framework. In the specific case of $K$-distributed noise, we show that, under certain conditions, the FIM for the deterministic part of the model can be unbounded, while the FIM for the covariance part of the model is always bounded. In the general case of elliptical distributions, we provide a sufficient condition for unboundedness of the  FIM. Accurate approximations of the FIM for $K$-distributed noise are also derived when it is bounded. Additionally, the maximum likelihood estimator of the signal DoA and an approximated version are derived, assuming known covariance matrix: the latter is then estimated from secondary data using a conventional regularization technique. When the FIM is unbounded, an analysis of the estimators reveals a rate of convergence much faster than the usual $T^{-1}$.  Simulations illustrate the different behaviors of the estimators, depending on the FIM being bounded or not.
\end{abstract}
\newpage

\section{Problem statement}
Estimating the direction of arrival (DoA) of multiple signals impinging on an array of sensors from observation of a finite number of array snapshots has been extensively studied in the literature \cite{VanTrees02}. Maximum likelihood estimators (MLE) and Cram\'{e}r-Rao bounds (CRB), derived under the assumption of additive white Gaussian noise, and either for the so-called conditional or unconditional model \cite{Bohme86,Stoica89,Stoica90,Stoica90b}, serve as references to which newly developed DoA estimators have been systematically compared.  In many instances however, additive noise is usually colored and, consequently, the problem of DoA estimation in spatially correlated noise fields has been studied, see e.g., \cite{Friedlander95,Ye96,Nagesha96,Viberg97,Goransson99}.  

When the spatial covariance matrix of this additive noise is known a priori, maximum likelihood estimators and Cram\'{e}r-Rao bounds are changing in a straightforward way with whitening operations. The new statistical problem appears when the covariance matrix of the additive noise is not known a priori and information about this matrix is substituted by a number of independent and identically distributed (i.i.d.) training samples, that form the so-called secondary training sample data set. In many cases one can assume that the statistical properties of the training noise data are the same as per noise data within the primary training set data: such conditions are usually referred to as the supervised training conditions. Therefore, under these conditions, one has two sets of measurements, one primary set $\Xp \in \mathbb{C}^{M \times \Tp}$ which contains signals of interest (SOI) and noise, and a second set $\Xs \in \mathbb{C}^{M \times \Ts}$ (secondary training set) which contains noise only. Examples of this problem formulation are numerous in the area of passive location and direction finding. For instance, in the so-called over-sampled 2D HF antenna arrays, ionospherically propagated external noise is spatially non white \cite{Coleman00,Abramovich13b}, and some parts of HF spectrum (distress signals for example) with no signals may be used for external noise sampling  \cite{Abramovich14}.  Despite its relevance in many practical situations, this problem has been relatively scarcely studied \cite{Abramovich04,Werner06}. For parametric description of the Gaussian noise covariance matrix with $\vthetan$  the unknown parameter vector, in \cite{Werner06}, the authors derive the Cram\'{e}r-Rao bound for joint  SOI parameters (DoA) $\vthetas$ and noise parameters $\vthetan$ estimation, assuming a conventional unconditional model, i.e., $\Xp \sim \mCN{\mat{0}}{\A(\vphi) \mGamma \A(\vphi)^{H} + \R_{n}}{\I_{\Tp}}$ and $\Xs \sim \mCN{\mat{0}}{\R_{n}}{\I_{\Ts}}$ where $\mCN{.}{.}{.}$ stands for the complex Gaussian distribution whose respective parameters are the mean, row covariance matrix and column covariance matrix. $\A(\vphi)$ is the usual steering matrix with $\vphi$  the vector of unknowns DoA,  $\mGamma$ denotes the waveforms covariance matrix and $\R_{n}$ corresponds to the noise covariance matrix, which is parameterized by vector $\vthetan$.

In many cases however, the Gaussian assumption for the predominant part of the noise cannot be advocated. Typical example is the HF external noise, heavily dominated by powerful lighting strikes \cite{George82,CCIR83,Radio13}. Evidence of deviations from the Gaussian assumption has been demonstrated numerous times for different applications, with  the relevance of the compound-Gaussian (CG) models being justified \cite{Conte87,Baker91,Rangaswamy93,Billingsley99,Conte04,Conte05}. In essence, the individual $M$-variate snapshot of such a noise over the face of an antenna array may be treated as a Gaussian random vector, whose power can randomly fluctuate from sample to sample. CG models belong to a larger class of distributions, namely multivariate elliptically contoured distributions (ECD) \cite{Anderson90,Fang90,Ollila12}.  For the sake of clarity, we briefly review the main definitions of ECD. A vector $\x \in \mathbb{C}^{M}$  follows an EC distribution if it admits the following stochastic representation
\begin{equation}\label{storep_CES}
\x \dist \m + \sqrt{\Q} \A \vu
\end{equation}
where $\dist$ means ``has the same distribution as''. In \eqref{storep_CES}, $\Q$  is a  non-negative real random variable  and is independent of the complex random vector $\vu$ which is  uniformly distributed over the complex sphere $\mathbb{C}S^{M} = \left\{ \vu \in \mathbb{C}^{M}; \left\| \vu \right\|=1 \right\}$. The matrix $\A$ is such that $\A \A^{H} = \R$ where $\R$ is the so-called scatter matrix, and we assume here that $\R$ is non-singular. The probability density function (p.d.f.) of $\x$ can then be written as
\begin{equation}\label{p(x)}
p(\x | \m, \R , g) \propto \det{\R}^{-1} g \left( (\x-\m)^{H} \iR (\x-\m) \right)
\end{equation}
where $\propto$ stands for proportional to.  The function $g: \, \mathbb{R}^{+} \, \longrightarrow \, \mathbb{R}^{+}$ is called the density generator and satisfies finite moment condition $\delta_{M,g} = \int_{0}^{\infty} t^{M-1} g(t) dt < \infty$. It is related to the p.d.f. of the modular variate $\Q$ by $p(\Q)= \delta_{M,g}^{-1} \Q^{M-1} g(\Q)$.

Going back to our scenario of two data sets $\Xp=\begin{bmatrix} \xp_{t_{1}} & \ldots & \xp_{\Tp} \end{bmatrix}$ and $ \Xs= \begin{bmatrix} \xs_{t_{1}} & \ldots & \xs_{\Ts} \end{bmatrix}$, we assume that they are independent, and that their columns are independent and distributed (i.i.d.) according to \eqref{p(x)}. In other words, one has $\xpt \dist \mt + \sqrt{\Qtp} \A \upt$ and $\xst \dist \sqrt{\Qts} \A \ust$, where $\Qtp$ and $\Qts$ are i.i.d. variables drawn from $p(\Q) \propto \Q^{M-1} g(\Q)$, and $\upt$ and $\ust$ are i.i.d. random vectors uniformly distributed on the unit sphere.   It then follows that the joint distribution of  $\Xp, \Xs$ is given by $p(\Xp , \Xs | \Mp, \R, g) = p(\Xp | \Mp, \R, g) p(\Xs | \R, g)$ where $\Mp = \begin{bmatrix} \m(1) & \ldots & \m(\Tp) \end{bmatrix}^{T}$ and
\begin{subequations}\label{p(Xp,Xs)}
\begin{align}
p(\Xp | \Mp, \R, g) &\propto \det{\R}^{-\Tp} \prod_{\tp=1}^{\Tp} g \left( \zpt^{H} \iR \zpt \right)  \\
p(\Xs | \R, g) &\propto \det{\R}^{-\Ts} \prod_{\ts=1}^{\Ts} g \left( \xst^{H} \iR \xst \right)
\end{align}
\end{subequations}
where $\zpt = \xpt-\mt$. Additionally, we assume that $\Mp$ depends on a parameter vector $\vthetas$ while $\R$ depends on $\vthetan$. Our objective is then to estimate $\vtheta = \begin{bmatrix} \vthetas \\ \vthetan \end{bmatrix}$ from $(\Xp , \Xs)$. Let us emphasize an essential difference of the problem in \eqref{p(Xp,Xs)} with respect to the typical problem of target detection in CG clutter \cite{Sangston12}. There, within each range resolution cell the clutter is perfectly Gaussian and therefore the optimum space-time processing is the same as per the standard Gaussian problem formulation. It is the data dependent threshold and clutter covariance matrix (in adaptive formulation) that needs to be calculated from the secondary data, if not known a priori \cite{Sangston12,Pascal08}. In the problem \eqref{p(Xp,Xs)}, the SOI DoA estimation should be performed on a number $\Tp$ of ECD i.i.d. primary training samples, and maximum likelihood DoA estimation algorithm and CRB should be expected to be very different from the Gaussian case. 

The paper is organized in the following way. In Section \ref{section:CRB}, we derive a general expression of the FIM for elliptically distributed noise using two data sets. Section \ref{section:doaK} focuses on the case of DoA estimation in $K$-distributed noise. In section \ref{section:crbK}, we derive conditions under which the FIM is bounded/unbounded, and provide a sufficient condition for unboundedness of the FIM with general elliptical distribution. The maximum likelihood estimate, as well as an approximation, are derived in section \ref{section:mleK}. In the same section, we derive  lower and  upper  bounds on the mean-square error of the MLE for non-regular estimation conditions, i.e., when the Fisher information matrix is unbounded. Numerical simulations serve to evaluate the performance of the estimators in Section \ref{section:numerical} and our conclusions are drawn in Section \ref{section:conclu}.

\section{Cram\'{e}r-Rao bounds\label{section:CRB}}
In this section, we derive the CRB for estimation of parameter vector $\vtheta$ from the distribution in \eqref{p(Xp,Xs)}. The Fisher information matrix (FIM) for the problem at hand can be written as \cite{VanTrees02}
\begin{align}\label{FIM}
\F(j,k)& = \E{\frac{\partial \log p(\Xp , \Xs | \Mp, \R, g)}{\partial \thetaj} \frac{\partial \log p(\Xp , \Xs | \Mp, \R, g)}{\partial \thetak}}  \nonumber \\
&=  \E{\frac{\partial \log p(\Xp  | \Mp, \R, g)}{\partial \thetaj} \frac{\partial \log p(\Xp | \Mp, \R, g)}{\partial \thetak}} \nonumber \\
& + \E{\frac{\partial \log p( \Xs | \R, g)}{\partial \thetaj} \frac{\partial \log p( \Xs |  \R, g)}{\partial \thetak}} \nonumber \\
\end{align}
where we used the fact that
\begin{equation}
\E{\frac{\partial \log p(\Xp  | \Mp, \R, g)}{\partial \thetaj} \frac{\partial \log p(\Xs | \R, g)}{\partial \thetak}} = 0.
\end{equation}
Hence, the total FIM is the sum of two matrices $\F = \Fp + \Fs$, with straightforward definition from \eqref{FIM}. In order to derive each matrix, we   will make use of the  general expression of the Fisher information matrix for ECD recently derived in \cite{Besson13,Greco13}. First, let us introduce
\begin{equation}
\alpha_{\mu} =  \E{\Qt^{\mu} \phi^{2}(\Qt)}
\end{equation}
where $\phi(t)=-\frac{g'(t)}{g(t)}$. Then, we have from \cite{Besson13} that the $(j,k)$-th element of the Fisher information matrices is given by
\begin{align}\label{F1}
\Fp(j,k)& = \frac{2 \alpha_{1}}{M} \sum_{\tp=1}^{\Tp}  \Real{\dmutdjH \iR \dmutdk} \nonumber \\
&+ \Tp \left[ \frac{\alpha_{2}}{M(M+1)} - 1 \right] \Tr{\iR \Rj} \Tr{\iR \Rk} \nonumber  \\
& + \frac{\alpha_{2} \Tp}{M(M+1)} \Tr{\iR \Rj \iR \Rk}
\end{align}
\begin{align}\label{F2}
\Fs(j,k) &= \Ts \left[ \frac{\alpha_{2}}{M(M+1)} - 1 \right] \Tr{\iR \Rj} \Tr{\iR \Rk} \nonumber  \\
& + \frac{\alpha_{2} \Ts}{M(M+1)} \Tr{\iR \Rj \iR \Rk}
\end{align}
where $\Rj  = \frac{\partial \R}{\partial \thetaj}$.  Since $\R$ depends only on $\vthetan$, it follows that $\Fs$ takes the following form
\begin{equation}
\Fs = \begin{bmatrix} \mat{0} & \mat{0} \\ \mat{0} & \Fs^{nn} \end{bmatrix}
\end{equation}
with 
\begin{align}
\Fs^{nn}(j,k) &=  \Ts \left[ \frac{\alpha_{2}}{M(M+1)} - 1 \right] \Tr{\iR \Rnj} \Tr{\iR \Rnj} \nonumber \\
&+ \frac{\alpha_{2} \Ts}{M(M+1)} \Tr{\iR \Rnj \iR \Rnk}
\end{align}
where $\Rnj  = \frac{\partial \R}{\partial \thetanj}$. Let us now consider $\Fp$. Using the fact that $\R$ depends only on $\vthetan$ and $\mt$ depends only on $\vthetas$, $\Fp$ is block-diagonal, i.e.,
\begin{equation}
\Fp = \begin{bmatrix} \Fp^{ss} & \mat{0} \\ \mat{0} & \Fp^{nn} \end{bmatrix}
\end{equation}
with
\begin{subequations}
\begin{align}
\Fp^{ss}(j,k) &= \frac{2 \alpha_{1}}{M} \sum_{\tp=1}^{\Tp}  \Real{\dmutdsjH \iR \dmutdsk}  \\
\Fp^{nn}(j,k) &= \Tp \left[ \frac{\alpha_{2}}{M(M+1)} - 1 \right] \Tr{\iR \Rnj} \Tr{\iR \Rnj} \nonumber \\
&+ \frac{\alpha_{2} \Tp}{M(M+1)} \Tr{\iR \Rnj \iR \Rnk}.
\end{align}
\end{subequations}
The whole FIM is thus given by
\begin{equation}
\F = \begin{bmatrix} \Fp^{ss} & \mat{0} \\ \mat{0} & \Fp^{nn} + \Fs^{nn} \end{bmatrix}
\end{equation}
The CRB for estimation of $\vthetas$ is obtained as the upper-left block of the inverse of the FIM and is thus simply $CRB(\vthetas) = \left( \Fp^{ss} \right)^{-1}$. Similarly to the Gaussian case, the CRB for estimation of $\vthetas$ in the conditional model is the same as if $\R$ was known. As for the CRB for estimation of $\vthetan$, it is the same as if we had a set of $T=\Tp + \Ts$ noise only samples.

\section{Application to $K$-distributed noise \label{section:doaK}}
\subsection{Data model}
We address the specific problem where the primary data can be written as
\begin{equation}\label{xt_CGK}
\xpt = \mt + \sqrt{\tau_{\tp}} \npt
\end{equation}
where $\tau_{\tp}$ follows a Gamma distribution with shape parameter $\nu$ and scale parameter $\beta$, i.e., its p.d.f. is given by
\begin{equation}
p(\tau_{\tp}) = \frac{\beta^{-\nu}}{\Gamma(\nu)} \tau_{\tp}^{\nu-1} e^{-\beta^{-1} \tau_{\tp}}.
\end{equation}
which we denote as $\tau_{t} \sim \pdfGamma{\nu}{\beta}$, and $\npt \sim \vCN{\vect{0}}{\R}$. The noise component is known to follow a $K$ distribution and $\xpt$ in \eqref{xt_CGK} admits a CES representation similar to \eqref{storep_CES} with $\Qtp \dist \pdfGamma{\nu}{\beta} \times \Cchisquare{M}$. The p.d.f. of $\Qtp$ in this case is given by
\begin{equation}\label{p(Qt)}
p(\Qtp) = \frac{2 \beta^{-(\nu+M)/2}}{\Gamma(\nu) \Gamma(M)} \Qtp^{\frac{\nu+M}{2}-1} K_{M-\nu}\left(2 \sqrt{\Qtp / \beta}\right)
\end{equation}
where $K_{M-\nu}(.)$ is the modified Bessel function. Note that the $\mu$-th order moment of $ \Qtp$ is 
\begin{align}\label{E{Q^mu}}
&\E{\Qtp^{\mu}} = \frac{2 \beta^{-(\nu+M)/2}}{\Gamma(\nu) \Gamma(M)} \int_{0}^{\infty} \Qtp^{\mu + \frac{\nu+M}{2}-1} K_{M-\nu}\left(2 \sqrt{\Qtp / \beta}\right) d \Qtp\nonumber \\
&= \frac{\beta^{\mu}}{2^{2\mu+\nu+M-2}\Gamma(\nu) \Gamma(M)} \int_{0}^{\infty} z^{2\mu+\nu+M-1} K_{M-\nu}\left(z\right) dz \nonumber \\
&=  \beta^{\mu} \frac{\Gamma(\mu+\nu) \Gamma(\mu+M)}{\Gamma(\nu) \Gamma(M)} \quad (\mu + \min(\nu,M) > 0)
\end{align}
where we used the fact \cite[6.656.16]{Gradshteyn94} that
\begin{equation}
\int_{0}^{\infty} z^{\mu} K_{\nu}(z) dz = \begin{cases} 2^{\mu-1} \Gamma\left(\frac{\mu+\nu+1}{2}\right) \Gamma\left(\frac{\mu-\nu+1}{2}\right) & \mu+1 \pm \nu > 0 \\
\infty & \text{otherwise} \end{cases}.
\end{equation}
The density generator is thus here
\begin{equation}\label{g_K}
g(\Q) =  \Q^{\frac{\nu-M}{2}} K_{M-\nu}\left(2 \sqrt{\Q / \beta}\right)
\end{equation}
where, for the sake of notational convenience, we have dropped the subscript $_{\tp}$. 

\subsection{Cram\'{e}r-Rao bounds\label{section:crbK}}
The FIM for $K$-distributed noise can be obtained from the FIM for Gaussian distributed noise and the calculation of the scalar
\begin{equation}\label{alpha_mu_ini}
\alpha_{\mu} = \E{\Q^{\mu} \left[ \frac{g'(\Q)}{g(\Q)} \right]^{2}}
\end{equation}
for $\mu \in \left\{1,2\right\}$. For the signal parameters part only, we indeed have $\F_{K}^{ss} = M^{-1} \alpha_{1} \F_{G}^{ss}$ where the subscript $_K$ and $_G$ stand for $K$-distributed and Gaussian distributed noise. Using the fact that $K'_{a}(z) = \frac{a}{z}K_{a}(z)  - K_{a+1}(z)$, it follows that
\begin{align}\label{g'(q)}
g'(\Q) &=  \frac{\nu-M}{2} \Q^{\frac{\nu-M}{2}-1} K_{M-\nu}\left(2 \sqrt{\Q / \beta}\right) + \Q^{\frac{\nu-M-1}{2}} \beta^{-1/2}  K'_{M-\nu}\left(2 \sqrt{\Q / \beta}\right) \nonumber \\
&= \frac{\nu-M}{2} \Q^{\frac{\nu-M}{2}-1} K_{M-\nu}\left(2 \sqrt{\Q / \beta}\right) \nonumber \\
&+ \Q^{\frac{\nu-M-1}{2}} \beta^{-1/2} \left[ \frac{M-\nu}{2 \sqrt{\Q / \beta}} K_{M-\nu}\left(2 \sqrt{\Q / \beta}\right) - K_{M+1-\nu}\left(2 \sqrt{\Q / \beta}\right) \right] \nonumber \\
&= - \Q^{\frac{\nu-M-1}{2}} \beta^{-1/2} K_{M+1-\nu}\left(2 \sqrt{\Q / \beta}\right).
\end{align}
It then ensues that
\begin{equation}\label{g'(q)/g(q)}
\phi(\Q) = -\frac{g'(\Q)}{g(\Q)} = \Q^{-1/2} \beta^{-1/2} \frac{K_{M+1-\nu}\left(2 \sqrt{\Q / \beta}\right)}{K_{M-\nu}\left(2 \sqrt{\Q / \beta}\right)}
\end{equation}
and thus
\begin{align}\label{alpha_mu}
&\alpha_{\mu} = \beta^{-1} \E{\Q^{\mu-1} \left[ \frac{K_{M+1-\nu}\left(2 \sqrt{\Q / \beta}\right)}{K_{M-\nu}\left(2 \sqrt{\Q / \beta}\right)} \right]^{2} } \nonumber \\
&=  \frac{2 \beta^{-\frac{\nu+M}{2}-1}}{\Gamma(\nu) \Gamma(M)} \int_{0}^{\infty} \Q^{\mu-2+ \frac{\nu+M}{2}} \frac{K^{2}_{M+1-\nu}\left(2 \sqrt{\Q / \beta}\right)}{K_{M-\nu}\left(2 \sqrt{\Q / \beta}\right)} d \Q \nonumber \\
&= \frac{\beta^{\mu-2}}{2^{2\mu+\nu+M-4}\Gamma(\nu) \Gamma(M)} \int_{0}^{\infty} z^{2\mu+\nu+M-3} \frac{K^{2}_{M+1-\nu}\left(z\right)}{K_{M-\nu}\left(z\right)} d z.
\end{align}
A formula for the FIM in case of $K$-distributed noise was derived in \cite{ElKorso14} based on the compound Gaussian representation \eqref{xt_CGK}. While it resembles our derivations based on the FIM for ECD derived in \cite{Besson13}, it does not match exactly our expression herein. Moreover,  we study herein the \emph{existence  of the FIM} and derive \emph{a closed-form approximation of the FIM}.

Let us investigate the conditions under which the integral
\begin{equation}\label{I_mu}
I_{\mu} = \int_{0}^{\infty} z^{2\mu+\nu+M-3} \frac{K^{2}_{M+1-\nu}\left(z\right)}{K_{M-\nu}\left(z\right)} d z
\end{equation}
converges. Towards this end, let us use the following inequality which holds for $M+1-\nu>1$  and $z>0$ \cite{Baricz10}
\begin{align}\label{lowerbound_besselk}
\frac{K_{M+1-\nu}(z)}{K_{M-\nu}(z)} &> \frac{(M+1-\nu) + \sqrt{\frac{M+1-\nu}{M-\nu}z^{2} + (M+1-\nu)^{2}}}{\frac{M+1-\nu}{M-\nu}z} \nonumber \\
& > \frac{(M+1-\nu) + \sqrt{\frac{M+1-\nu}{M-\nu}}z}{\frac{M+1-\nu}{M-\nu}z} \nonumber \\
&= \left(\frac{M-\nu}{M+1-\nu}\right)^{1/2} + (M-\nu) z^{-1}.
\end{align}
It follows that
\begin{align}
I_{\mu} & > \left(\frac{M-\nu}{M+1-\nu}\right)^{1/2} \int_{0}^{\infty} z^{2\mu+\nu+M-3} K_{M+1-\nu}\left(z\right) d z  \nonumber \\
&+ (M-\nu) \int_{0}^{\infty} z^{2\mu+\nu+M-4} K_{M+1-\nu}\left(z\right) d z
\end{align}
The first integral converges for $2\mu+\nu+M-2-M-1+\nu > 0 \Leftrightarrow \mu+\nu-\frac{3}{2}>0$ while the second converges for $2\mu+\nu+M-3-M-1+\nu > 0 \Leftrightarrow \mu+\nu-2>0$. Hence, for $\mu+\nu-2>0$, one has
\begin{align}
I_{\mu} &> \left(\frac{M-\nu}{M+1-\nu}\right)^{1/2} 2^{2\mu+\nu+M-4} \Gamma(\mu+M-\frac{1}{2}) \Gamma(\mu+\nu-\frac{3}{2}) \nonumber \\
&+ (M-\nu) 2^{2\mu+\nu+M-5} \Gamma(\mu+M-1) \Gamma(\mu+\nu-2).
\end{align}
Accordingly, one has, for $z>0$
\begin{align}\label{upperbound_besselk}
\frac{K_{M+1-\nu}(z)}{K_{M-\nu}(z)}  &< \frac{(M+1-\nu) + \sqrt{z^{2} + (M+1-\nu)^{2}}}{z}  \nonumber \\
&< 2(M+1-\nu) z^{-1} + 1
\end{align}
which implies that
\begin{align}
I_{\mu} & < 2(M+1-\nu) \int_{0}^{\infty} z^{2\mu+\nu+M-4} K_{M+1-\nu}\left(z\right) d z  \nonumber \\
&+ \int_{0}^{\infty} z^{2\mu+\nu+M-3} K_{M+1-\nu}\left(z\right).
\end{align}
The first integral converges for $\mu+\nu-2>0$ and the second converges for $\mu+\nu-\frac{3}{2}>0$. In the former case, one has
\begin{align}
I_{\mu} & <  (M+1-\nu) 2^{2\mu+\nu+M-4} \Gamma(\mu+M-1) \Gamma(\mu+\nu-2)\nonumber \\
&+  2^{2\mu+\nu+M-4} \Gamma(\mu+M-\frac{1}{2}) \Gamma(\mu+\nu-\frac{3}{2}).
\end{align}
Consequently, we conclude that \emph{the integral converges only for} $\mu+\nu-2>0$: for $\mu=2$ this implies that $\nu>0$ which is verified. In contrast, when $\mu=1$, one must have $\nu > 1$. In other words, the term in \emph{the FIM corresponding to the noise parameters is always bounded} since it depends on $I_{2}$ only. The situation is different for signal parameters. In an unconditional model where $\R$ would depend on signal parameters as well, the FIM is bounded. In contrast, in the conditional model where signal parameters are embedded in the mean of the distribution, the FIM corresponding to signal parameters\emph{ is bounded only for $\nu > 1$}: otherwise, it is unbounded. The latter case corresponds to the so-called non regular case corresponding to distributions with singularities, as studied e.g., in \cite{Ibragimov81}. 

Before pursuing our study of the FIM for the specific case of $K$-distributed noise, let us make an important observation. For the $K$ distribution, we have just proven that  $I_{1}$ does not exist for $\nu \leq 1$. However,  see \eqref{E{Q^mu}},  $\E{\Qtp^{\mu}}$ exists if and only if $\mu+M > 0$ and $\mu+\nu>0$. The latter condition implies that, when $\nu \leq 1$, $\E{\Qtp^{-1}}=\delta_{M,g}^{-1}  \int_{0}^{\infty} \Q^{M-2} g(\Q) d \Q$ does not exist. Observe that convergence of the latter integral is problematic in a neighborhood of $0$, since for $\Q_{0} > 1$, $\int_{\Q_{0}}^{b}\Q^{M-2} g(\Q) d \Q <  \int_{\Q_{0}}^{b}\Q^{M-1} g(\Q) d \Q < \delta_{M,g}$ as $p(.)$ is a density.  Therefore, at least for $K$-distributed noise, if $\E{\Qtp^{-1}}$ does not exist, then $\E{\Qtp \phi^{2}(\Qtp)}$ is unbounded. At this stage, one may wonder if this property extends to any other elliptical distribution. It turns out that  this is indeed the case, as stated and proved in the next proposition.

\begin{prop}
Whatever the p.d.f. of the modular variate $\Qtp$, if $\E{\Qtp^{-1}} = \infty$ then $\E{\Qtp \phi^{2}(\Qtp)} = \infty$.
\end{prop}
\begin{proof}
For the sake of notational convenience, we temporarily omit the subscript $_{\tp}$ and use $\Q$ instead of $\Qtp$. Let us first observe that
\begin{equation}
\E{\Q \phi^{2}(\Q)} =  \int_{0}^{\infty} \Q \phi^{2}(\Q) p(\Q) d\Q =  \int_{0}^{\infty} \Q^{-1} \psi^{2}(\Q) p(\Q) d\Q.
\end{equation}
Since $p(\Q) = \delta_{M,g}^{-1} \Q^{M-1} g(\Q)$, one can write
\begin{align}
\psi(\Q) &= - \Q \frac{g'(\Q)}{g(Q)} = -\Q \frac{\partial \ln g(\Q)}{\partial \Q} \nonumber \\
&= -\Q \frac{\partial \ln \Q^{1-M} p(\Q)}{\partial \Q}  \nonumber \\
&= - (1-M) \Q \frac{\partial \ln \Q}{\partial \Q}  - \Q \frac{\partial \ln  p(\Q)}{\partial \Q}  \nonumber \\
&= (M-1) - \Q \frac{\partial \ln  p(\Q)}{\partial \Q}.
\end{align}
which implies that
\begin{align}
\Q^{-1} \psi^{2}(\Q)&= (M-1)^{2} \Q^{-1} - 2 (M-1) \frac{\partial \ln  p(\Q)}{\partial \Q}  \nonumber \\
&+ \Q \left[ \frac{\partial \ln  p(\Q)}{\partial \Q} \right]^{2}.
\end{align}
Therefore
\begin{align}
\int_{a}^{b}\Q \phi^{2}(\Q) p(\Q) d\Q &= (M-1)^{2} \int_{a}^{b} \Q^{-1} p(\Q) d\Q - 2 (M-1) \int_{a}^{b}  \frac{\partial \ln  p(\Q)}{\partial \Q} p(\Q) d\Q\nonumber \\
&+ \int_{a}^{b}  \Q \left[ \frac{\partial \ln  p(\Q)}{\partial \Q} \right]^{2} p(\Q) d\Q \nonumber \\
&= (M-1)^{2} \int_{a}^{b} \Q^{-1} p(\Q) d\Q - 2 (M-1) \int_{a}^{b}   p'(\Q) d\Q\nonumber \\
&+ \int_{a}^{b}  \Q \left[ \frac{\partial \ln  p(\Q)}{\partial \Q} \right]^{2} p(\Q) d\Q  \nonumber \\
&= (M-1)^{2} \int_{a}^{b} \Q^{-1} p(\Q) d\Q - 2  (M-1) \left[ p(b) - p(a) \right] \nonumber \\
&+ \int_{a}^{b}  \Q \left[ \frac{\partial \ln  p(\Q)}{\partial \Q} \right]^{2} p(\Q) d\Q. 
\end{align}
The third term of the sum is always positive. In the second term, we have that $\lim_{b \rightarrow \infty} p(b) = 0$. It follows that divergence of $\int_{a}^{b} \Q^{-1} p(\Q) d\Q$ is a sufficient condition for divergence of $\int_{a}^{b}\Q \phi^{2}(\Q) p(\Q) d\Q$. As said before $\lim_{b \rightarrow \infty} \int_{a}^{b} \Q^{-1} p(\Q) d\Q$ exists, and therefore a sufficient condition for $\E{\Q \phi^{2}(\Q)}$ to be undounded is that $\lim_{a \rightarrow 0} \int_{a}^{\infty} \Q^{-1} p(\Q) d\Q = \E{\Q^{-1}}$ is unbounded.
\end{proof}

Let us now go back to the $K$-distributed case and  investigate whether it is possible to derive a simple expression for $I_{\mu}$ and subsequently $\alpha_{\mu}$, assuming that $\mu+\nu-2>0$. Towards this end, let us make use of
\begin{equation}
K_{M+1-\nu}(z) = \frac{2(M-\nu)}{z} K_{M-\nu}(z) + K_{M-1-\nu}(z)
\end{equation}
to write that
\begin{align}
I_{\mu} &= \int_{0}^{\infty} z^{2\mu+\nu+M-3} \frac{K^{2}_{M+1-\nu}\left(z\right)}{K_{M-\nu}\left(z\right)} d z \nonumber \\
&=4(M-\nu)^{2} \int_{0}^{\infty} z^{2\mu+\nu+M-5} K_{M-\nu}(z)  d z \nonumber \\
&+4(M-\nu) \int_{0}^{\infty} z^{2\mu+\nu+M-4} K_{M-1-\nu}(z) d z \nonumber \\
&+ \int_{0}^{\infty} z^{2\mu+\nu+M-3} \frac{K^{2}_{M-1-\nu}\left(z\right)}{K_{M-\nu}\left(z\right)} d z \nonumber \\
&= 2^{2\mu+\nu+M-4} (M-\nu)^{2} \Gamma(\mu+M-2) \Gamma(\mu+\nu-2) \nonumber \\
& + 2^{2\mu+\nu+M-3} (M-\nu) \Gamma(\mu+M-2) \Gamma(\mu+\nu-1) \nonumber \\
&+ \int_{0}^{\infty} z^{2\mu+\nu+M-3} \frac{K^{2}_{M-1-\nu}\left(z\right)}{K_{M-\nu}\left(z\right)} d z
\end{align} 
The last term is obviously not possible to obtain in closed-form so that we use a ``large $M-\nu$'' approximation of the modified Bessel function \cite{NIST10}
\begin{equation}\label{approxK_largeM}
K_{M-\nu}(z) \simeq \sqrt{\frac{\pi}{2(M-\nu)}} \left( \frac{ez}{2(M-\nu)}\right)^{-(M-\nu)}
\end{equation}
which results in
\begin{align}
\frac{K_{M-1-\nu}\left(z\right)}{K_{M-\nu}\left(z\right)} & \simeq \left( \frac{M-\nu}{M-1-\nu}\right)^{1/2} \frac{(M-1-\nu)^{M-1-\nu)}}{(M-\nu)^{M-\nu)}}  \frac{ez}{2}\nonumber \\
& \simeq \left( \frac{M-\nu}{M-1-\nu}\right)^{1/2} \frac{1}{e(M-\nu)}  \frac{ez}{2} \nonumber \\
&= \frac{z}{2(M-\nu)^{1/2} (M-1-\nu)^{1/2}}.
\end{align}
Therefore, 
\begin{align}
\int_{0}^{\infty} z^{2\mu+\nu+M-3} \frac{K^{2}_{M-1-\nu}\left(z\right)}{K_{M-\nu}\left(z\right)} d z & \simeq \frac{1}{2(M-\nu)^{1/2} (M-1-\nu)^{1/2}}  \int_{0}^{\infty} z^{2\mu+\nu+M-2} K_{M-1-\nu}\left(z\right) d z\nonumber \\
&= \frac{ 2^{2\mu+\nu+M-4} }{(M-\nu)^{1/2} (M-1-\nu)^{1/2}}\Gamma(\mu+M-1) \Gamma(\mu+\nu)
\end{align}
We finally have
\begin{align}\label{alpha_mu_approx_3terms}
&\alpha_{\mu} = \frac{\beta^{\mu-2} I_{\mu}}{2^{2\mu+\nu+M-4}\Gamma(M) \Gamma(\nu)}   \nonumber \\
& \simeq  \beta^{\mu-2} (M-\nu)^{2} \frac{\Gamma(\mu+M-2) \Gamma(\mu+\nu-2)}{\Gamma(M) \Gamma(\nu)} \nonumber \\
&+ 2 \beta^{\mu-2} (M-\nu) \frac{\Gamma(\mu+M-2) \Gamma(\mu+\nu-1)}{\Gamma(M) \Gamma(\nu)} \nonumber \\
&+ \beta^{\mu-2} (M-\nu)^{-1/2} (M-1-\nu)^{-1/2} \frac{\Gamma(\mu+M-1) \Gamma(\mu+\nu)}{\Gamma(M) \Gamma(\nu)}.
\end{align}

If the large $M-\nu$ approximation is made from the start, then one has
\begin{align}\label{approx_ratioK}
\frac{K_{M+1-\nu}\left(z\right)}{K_{M-\nu}\left(z\right)} & \simeq  2(M+1-\nu)^{1/2} (M-\nu)^{1/2} z^{-1}
\end{align}
so that
\begin{align}
I_{\mu}  &\simeq 2(M+1-\nu)^{1/2} (M-\nu)^{1/2} \int_{0}^{\infty} z^{2\mu+\nu+M-4} K_{M+1-\nu}\left(z\right) d z  \nonumber \\
&= 2^{2\mu+\nu+M-4} (M+1-\nu)^{1/2} (M-\nu)^{1/2} \Gamma(\mu+M-1) \Gamma(\mu+\nu-2)
\end{align}
and hence
\begin{equation}\label{alpha_mu_approx_1term}
\alpha_{\mu} \simeq \beta^{\mu-2} (M+1-\nu)^{1/2} (M-\nu)^{1/2} \frac{\Gamma(\mu+M-1) \Gamma(\mu+\nu-2)}{\Gamma(M) \Gamma(\nu)}
\end{equation}

Figure \ref{fig:alphamu_approx} compares the approximations in  \eqref{alpha_mu_approx_3terms} and \eqref{alpha_mu_approx_1term}, as well as a method which uses random number generation to approximate $\alpha_{\mu}$ based on its initial definition in \eqref{alpha_mu_ini}. More precisely, we generated a large number of random variables $\Q \dist \pdfGamma{\nu}{\beta} \times \Cchisquare{M}$ and replace the statistical expectation of  \eqref{alpha_mu_ini} by an average over the so-generated random variables. As can be observed from Figure \ref{fig:alphamu_approx}, the $3$ approximations provide very close values, which enable one to validate the closed-form expressions in \eqref{alpha_mu_approx_3terms} and \eqref{alpha_mu_approx_1term}.
\begin{figure}[htb]
\centering
\includegraphics[width=7.5cm]{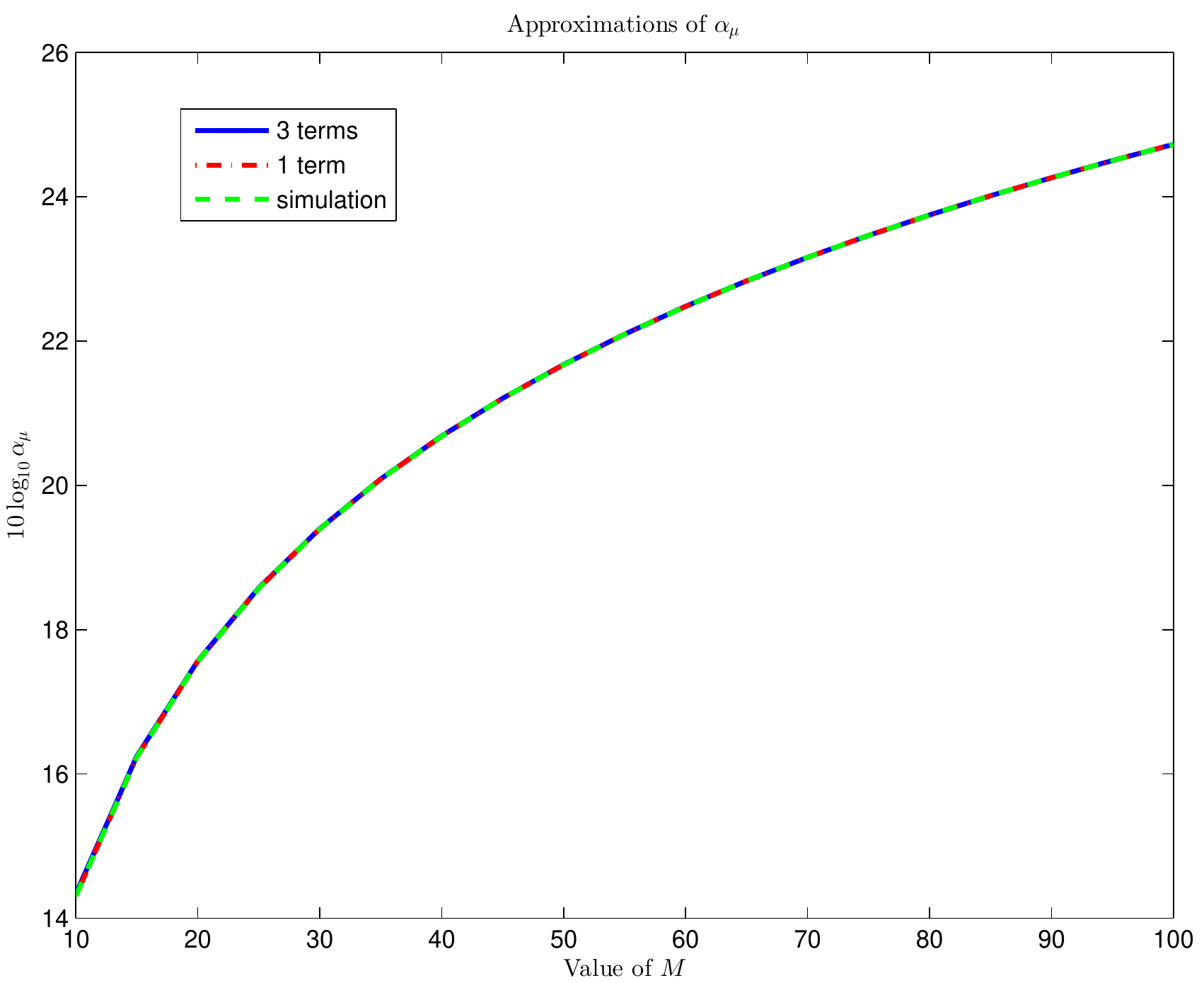}
\caption{Comparison of the approximations of $\alpha_{\mu}$ in \eqref{alpha_mu_approx_3terms} and \eqref{alpha_mu_approx_1term}.  $\nu=1.5$ and $\beta=1/\nu$.}
\label{fig:alphamu_approx}
\end{figure}

\subsection{Maximum Likelihood estimation\label{section:mleK}}
We now focus on maximum likelihood (ML) estimation of direction of arrival $\phitrue$, signal waveforms $s_{\tp}$ and covariance matrix $\R$ in the model
\begin{align}\label{modelxpt_K}
\xpt &= \va(\phitrue) s_{\tp} + \sqrt{\tau_{\tp}} \npt; \quad \tp=1,\ldots,\Tp \nonumber \\
\xst &=  \sqrt{\tau_{\ts}} \nst; \quad \ts=1,\ldots,\Ts
\end{align}
where  $\tau_{\tp}, \tau_{\ts} \sim \pdfGamma{\nu}{\beta}$, and $\npt , \nst \sim \vCN{\vect{0}}{\R}$. The joint distribution of $(\Xp,\Xs)$ is given by
\begin{align}
p(\Xp,\Xs) &\propto \det{\R}^{-(\Tp+\Ts)} \prod_{\ts=1}^{\Ts} \left[ \xst^{H} \iR \xst \right]^{\frac{\nu-M}{2}} K_{M-\nu}\left( 2 \sqrt{\xst^{H} \iR \xst / \beta} \right) \nonumber \\
&\times \prod_{\tp=1}^{\Tp} \left[ \zpt^{H} \iR\zpt \right]^{\frac{\nu-M}{2}} K_{M-\nu}\left( 2 \sqrt{ \zpt^{H} \iR\zpt / \beta }\right)
\end{align}
where $\zpt=\xpt- \va(\phi) s_{\tp}$. Joint estimation of all parameters appears to be very complicated and hence we will proceed in two steps. At first, we assume that $\R$ is known and derive the ML estimates of $\phi$ and $s_{\tp}$. Then, $\R$ is substituted for some estimate obtained from observation of $\Xs$ only.

\subsubsection{DoA estimation with known $\R$}
Assuming that $\R$ is known, one needs to maximize with respect to $\phi$ and $s_{\tp}$
\begin{equation}
p(\Xp) \propto \prod_{\tp=1}^{\Tp} g \left( \left[\xpt- \va(\phi) s_{\tp}\right]^{H} \iR \left[\xpt- \va(\phi) s_{\tp}\right] \right)
\end{equation}
where $g(.)$ is given by \eqref{g_K}. Since $g(.)$ is monotonically decreasing, see \eqref{g'(q)}, it follows that $p(\Xp)$ is maximized when the argument of $g(.)$ is minimized. However,
\begin{align}
&\left[\xpt- \va(\phi) s_{\tp}\right]^{H} \iR \left[\xpt- \va(\phi) s_{\tp}\right] \nonumber \\
& = \left[\va^{H}(\phi) \iR \va(\phi)  \right] \left| s_{\tp} - \frac{\va^{H}(\phi) \iR \xpt}{\va^{H}(\phi) \iR \va(\phi)} \right|^{2} \nonumber \\
&+ \xpt^{H} \iR \xpt - \frac{\left| \va^{H}(\phi) \iR \xpt \right|^{2}}{\va^{H}(\phi) \iR \va(\phi)}.
\end{align}
Therefore, for any $\phi$, $p(\Xp)$ is maximized when
\begin{equation}
s_{\tp} = \frac{\va^{H}(\phi) \iR \xpt}{\va^{H}(\phi) \iR \va(\phi)}.
\end{equation}
It ensues that one needs now to maximize, with respect to $\phi$
\begin{equation}\label{f(phi)}
f(\phi) = \prod_{\tp=1}^{\Tp} g \left( \xpt^{H} \iR \xpt - \frac{\left| \va^{H}(\phi) \iR \xpt \right|^{2}}{\va^{H}(\phi) \iR \va(\phi)} \right).
\end{equation}
with $g(z) = z^{\frac{\nu-M}{2}} K_{M-\nu}(2\sqrt{z/ \beta})$. In order to avoid calculation of a modified Bessel function and thus in order to simplify estimation, we propose to make use of the ``large $M-\nu$'' approximation of the modified Bessel function given in \eqref{approxK_largeM} to write
\begin{align}
g(z) &= z^{\frac{\nu-M}{2}} K_{M-\nu}(2\sqrt{z/ \beta}) \nonumber \\
&\simeq z^{\frac{\nu-M}{2}} \times \sqrt{\frac{\pi}{2(M-\nu)}} \left( \frac{e\sqrt{z/ \beta}}{(M-\nu)}\right)^{-(M-\nu)} \nonumber \\
&= \mathrm{const.} z^{\nu-M}.
\end{align}
This approximation results in an approximate maximum likelihood  (AML) estimator of $\phi$ which consists in maximizing
\begin{equation}\label{ftilde(phi)}
\tilde{f}(\phi) = \prod_{\tp=1}^{\Tp} \left[ \xpt^{H} \iR \xpt - \frac{\left| \va^{H}(\phi) \iR \xpt \right|^{2}}{\va^{H}(\phi) \iR \va(\phi)} \right]^{\nu-M}.
\end{equation}
Note that
\begin{equation}\label{log_ftilde(phi)}
\log \tilde{f}(\phi) = (\nu-M) \sum_{\tp=1}^{\Tp} \log \left[ \xpt^{H} \iR \xpt - \frac{\left| \va^{H}(\phi) \iR \xpt \right|^{2}}{\va^{H}(\phi) \iR \va(\phi)} \right]
\end{equation}
which should be compared to the concentrated log likelihood function in the Gaussian case, as given by
\begin{equation}\label{log_f_G(phi)}
\log f_{G}(\phi) =  - \sum_{\tp=1}^{\Tp}  \left[ \xpt^{H} \iR \xpt - \frac{\left| \va^{H}(\phi) \iR \xpt \right|^{2}}{\va^{H}(\phi) \iR \va(\phi)} \right].
\end{equation}

A few remarks are in order about these estimates, in particular about the behavior of the AML estimator in the case of unbounded FIM, i.e., when $0 < \nu < 1$ .  First, note that all estimates will be a function of
\begin{align}\label{t(xt,phi)}
t(\xpt,\phi) &= \xpt^{H} \iR \xpt - \frac{\left| \va^{H}(\phi) \iR \xpt \right|^{2}}{\va^{H}(\phi) \iR \va(\phi)} \nonumber \\
&= \xpt^{H} \isqrtR \Porthaw \isqrtR \xpt
\end{align}
where  $\Porthaw$ is the projection onto the orthogonal complement of $\isqrtR \va(\phi)$.  Compared to  \eqref{log_f_G(phi)}, the logarithm operation in \eqref{ftilde(phi)} will strongly emphasize those snapshots $\xpt$  for which $t(\xpt,\phi)$ is small. Let us thus investigate the properties of this statistic, when evaluated at the \emph{true} value of signal DOA $\phitrue$. Using the fact that $\isqrtR \xpt = \isqrtR \atrue s_{\tp} + \sqrt{\tau_{\tp}} \wpt$, where $\wpt \sim \vCN{\vect{0}}{\I_{M}}$ and  $\atrue$ is a short-hand notation for $\va(\phitrue)$, one has 
\begin{equation}
t(\xpt,\phitrue) = \tau_{\tp} \wpt^{H} \Porthawtrue \wpt \dist \tau_{\tp} \times \Cchisquare{M-1}.
\end{equation}
For small $\nu$ ($0 < \nu < 1$), it follows that, in the vicinity of $\phitrue$, the snapshot with minimal $t(\xpt,\phi)$ is more or less the snapshot for which $\tau_{\tp}$ is minimum, hence the snapshot for which noise power is minimum, which makes sense. If we let $u_{\Tp} = \min_{1 \leq \tp \leq \Tp} \tau_{\tp}$, then its cumulative density function (c.d.f.)  is given by  
\begin{align}\label{cdf_min_tau}
\Prob{u_{\Tp} \leq \eta} &= 1 - \left( 1 - \Prob{\tau_{\tp} \leq \eta} \right)^{\Tp} \nonumber \\
&= 1 - \left[ 1 - \gamma\left(\nu, \eta \beta^{-1}\right) \right]^{\Tp} 
\end{align}
which is shown in Figure \ref{fig:cdf_min_tau_T=4}. Obviously, with small $\nu$, the snapshot which corresponds to the minimum value of $\tau_{\tp}$ exhibits a very high signal to noise ratio and, due to the emphasizing effect of the $\log$ operation in \eqref{ftilde(phi)}, the performance of the AML estimator is likely to be driven mainly by this particular snapshot. This is illustrated in Figure \ref{fig:MSE_Rknown_order_tau_vs_T_vs_nu} where we display the mean-square error (MSE) of the AML estimate which uses all $\Tp$ snapshots and the MSE of an hypothetical AML estimator which would use only the snapshot $\x_{t_{\min}}$ corresponding to the minimum value of $\tau_{\tp}$. The scenario of this simulation is described in the next section. This figure shows a marginal loss of the AML estimator using $\x_{t_{\min}}$ only, as compared to the full AML estimator, especially for small $\nu$.
\begin{figure}[htb]
\centering
\includegraphics[width=7.5cm]{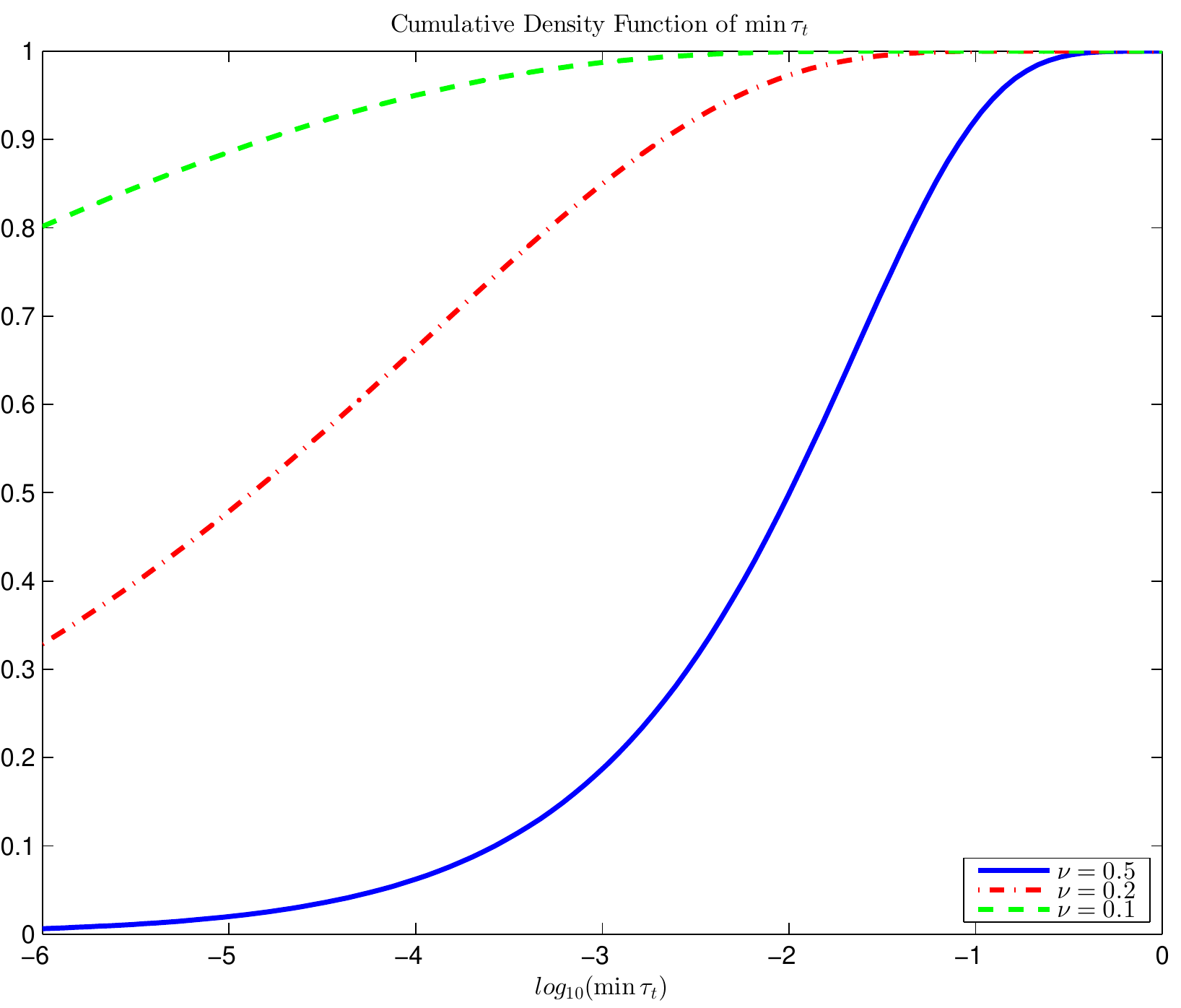}
\caption{Cumulative density function of $\min_{1 \leq \tp \leq \Tp} \tau_{\tp}$. $M=16$ and $\Tp=4$.}
\label{fig:cdf_min_tau_T=4}
\end{figure}
\begin{figure}[htb]
\centering
\includegraphics[width=7.5cm]{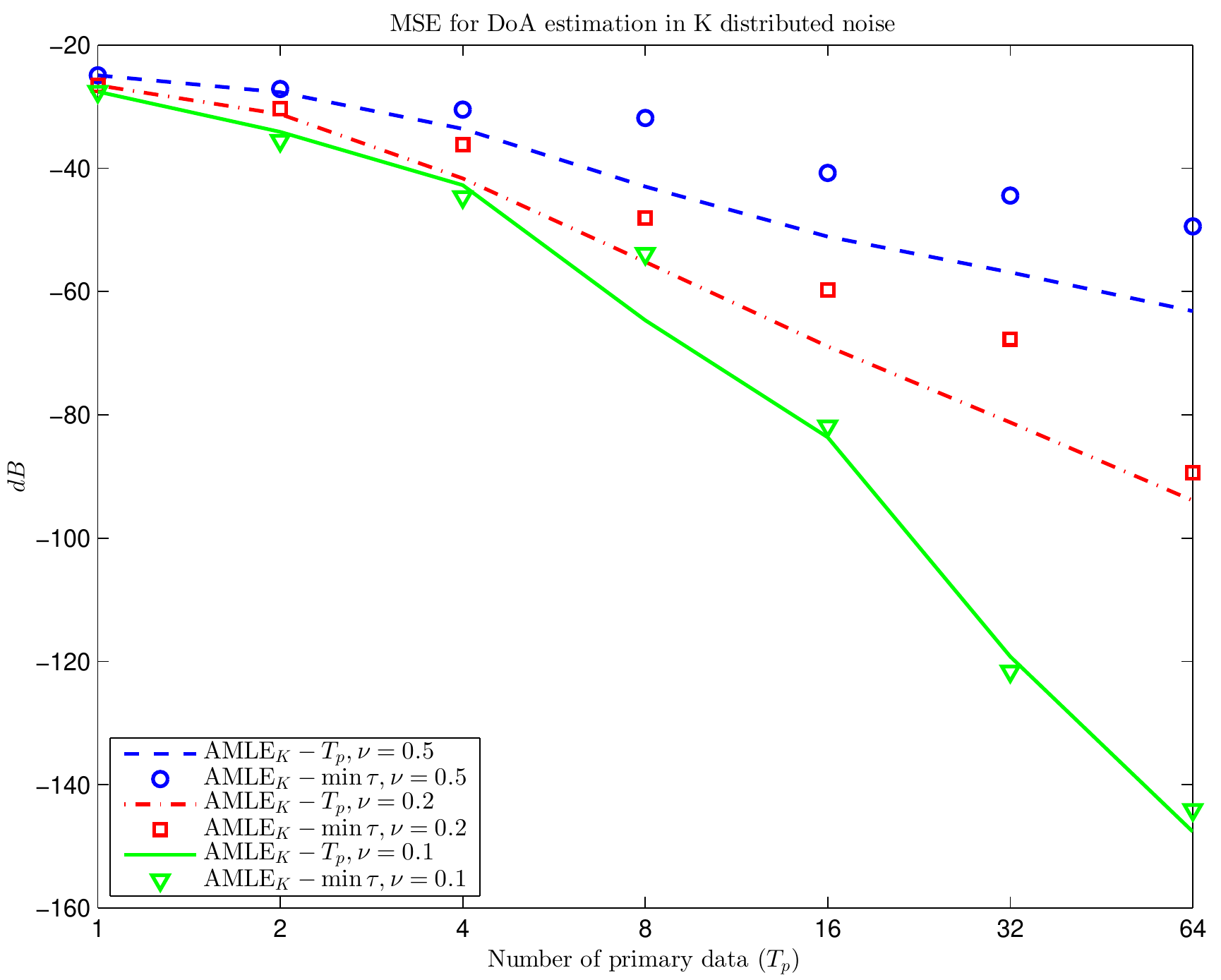}
\caption{Mean square error of  AML estimator using either all snapshots or a single snapshot  corresponding to minimal $\tau_{\tp}$. $\R$ known,  $M=16$ and $SNR=3$dB.}
\label{fig:MSE_Rknown_order_tau_vs_T_vs_nu}
\end{figure}

Let us thus analyze the behavior of the AML estimators. For the sake of notational convenience, let $\phi_{K}^{\Tp}$ and  $\phi_{K}^{\min}$ denote the AML estimator using $\Tp$ snapshots with $K$-distributed noise and the AML estimator using the snapshot $\x_{t_{\min}}$ corresponding to the minimal $\tau_{\tp}$, respectively.  Observe that, when using a single snapshot $\x_{t_{\min}}$,  minimizing \eqref{ftilde(phi)} is equivalent to minimizing the Gaussian likelihood function in \eqref{log_f_G(phi)} with $\Tp=1$. Since $\x_{t_{\min}}$ exhibits a high signal to noise ratio,  $\phi_{K}^{\min}$ is close to $\phitrue$, one can make a Taylor expansion and relate the error $\phi_{K}^{\min} - \phitrue$ to the error $\x_{t_{\min}}-\atrue s_{t_{\min}}$ as 
\begin{equation}
\phi_{K}^{\min} - \phitrue \simeq \sqrt{u_{\Tp}} \vzeta^{H} \n_{t_{\min}} 
\end{equation}
where $\vzeta$ is some vector that depends essentially on the derivatives of $\va(\phi)$ \cite{Renaux06} and whose expression is not needed here. One can simply notice that $\vzeta$ would be the same with Gaussian noise and a single snapshot, since maximizing \eqref{ftilde(phi)} or \eqref{log_f_G(phi)} is equivalent when one snapshot is used . This implies that
\begin{equation}\label{aMSE_AMLEK_min_tau}
\E{\left( \phi_{K}^{\min} - \phitrue\right)^{2}} \simeq \E{u_{\Tp}}  \vzeta^{H} \R \vzeta.
\end{equation}
Observe that $\vzeta^{H} \R \vzeta$ is the mean-square error (MSE) that would obtained in Gaussian noise and a single snapshot, which is about $\Tp$ times the MSE obtained in the Gaussian case and using $\Tp$ snapshots, and the latter is approximately the Gaussian CRB. The MSE of $\phi_{K}^{\min}$ depends on  $\E{u_{\Tp}}$ where $u_{\Tp}$ is the minimum value of a set of $\Tp$ independent and identically distributed (actually gamma distributed) variables. Therefore, in order to obtain $\E{u_{\Tp}}$, one must consider statistics of extreme values, a field that has received considerable attention for a long time, see e.g.,  \cite{Gumbel35,Gnedenko43,Gumbel58}. It turns out that only asymptotic (as $\Tp \rightarrow \infty$) results are available and we build upon them to derive the rate of convergence of $\E{\left( \phi_{K}^{\min} - \phitrue\right)^{2}}$. First, note that
\begin{align}
\Prob{\Tp^{1/ \nu} u_{\Tp} \geq x} &=  \Prob{u_{\Tp} \geq \Tp^{-1/ \nu} x} \nonumber \\
&= \left( \Prob{\tau_{\tp} \geq \Tp^{-1/ \nu} x} \right)^{\Tp} \nonumber \\
&= \left(1 - \Prob{\tau_{\tp} \leq \Tp^{-1/ \nu} x} \right)^{\Tp} \nonumber \\
&=\left[ 1 - \gamma\left(\nu,\beta^{-1}\Tp^{-1/ \nu} x\right) \right]^{\Tp}.
\end{align}
Now since $\nu$ is small and $\Tp$ is large, $\Tp^{-1/ \nu}$ is very small and we can approximate $\gamma(a,y) \simeq \left[ a \Gamma(a) \right]^{-1} y^{a}$, which yields
\begin{align}\label{cdf_vTp}
\Prob{\Tp^{1/ \nu} u_{\Tp} \geq x} &=\left[ 1 - \gamma\left(\nu,\beta^{-1}\Tp^{-1/ \nu} x\right) \right]^{\Tp} \nonumber \\
&\simeq \left[ 1 - \frac{\beta^{-\nu} \Tp^{-1} x^{\nu}}{\nu \Gamma(\nu)} \right]^{\Tp} \nonumber \\
&\simeq \ex{-\frac{\beta^{-\nu} \Tp^{-1} x^{\nu}}{\nu \Gamma(\nu)}}^{\Tp} \nonumber \\
&= \ex{-\frac{\beta^{-\nu} x^{\nu}}{\nu \Gamma(\nu)}}.
\end{align}
It follows that asymptotically, $v_{\Tp} = \Tp^{1/ \nu} u_{\Tp}$ converges to the distribution in \eqref{cdf_vTp}, whose probability density function is  
\begin{equation}
p(v_{\Tp}) = \frac{\beta^{-\nu}}{\Gamma(\nu)} v_{\Tp}^{\nu-1} \ex{-\frac{\beta^{-\nu} v_{\Tp}^{\nu}}{\nu \Gamma(\nu)}}.
\end{equation}
Using integration by parts, it follows that
\begin{align}
&\E{v_{\Tp}} = \int_{0}^{\infty} \frac{\beta^{-\nu}}{\Gamma(\nu)} x^{\nu} \ex{-\frac{\beta^{-\nu} x^{\nu}}{\nu \Gamma(\nu)}} dx \nonumber \\
&= \left[ -x \ex{-\frac{\beta^{-\nu} x^{\nu}}{\nu \Gamma(\nu)}} \right]^{\infty}_{0} + \int_{0}^{\infty} \ex{-\frac{\beta^{-\nu} x^{\nu}}{\nu \Gamma(\nu)}} dx \nonumber \\
&= \beta \nu^{1/\nu -1} \Gamma(\nu)^{1/\nu} \int_{0}^{\infty} z^{1/\nu -1} \ex{-z} dz \nonumber \\
&= \beta \nu^{1/\nu -1} \Gamma(\nu)^{1/\nu} \Gamma(\nu^{-1}) \triangleq C(\nu,\beta).
\end{align}
One can then conclude that, as $\Tp$ goes to infinity,
\begin{equation}
\E{\left( \phi_{K}^{\min} - \phitrue\right)^{2}} \simeq  \left[\vzeta^{H} \R \vzeta\right] C(\nu,\beta) \Tp^{-1/\nu}.
\end{equation}
Therefore, in the case of $0 < \nu < 1$, the MSE of $\phi_{K}^{\min}$ decreases as $\Tp^{-1/\nu}$, a rate of convergence much faster than the usual $\Tp^{-1}$. Note that this case corresponds to unbounded FIM. Such rates of convergence are also found with distributions possessing singularities \cite[chapter 6]{Ibragimov81}.

As for the AML estimate obtained from $\Tp$ snapshots, namely $\phi_{K}^{\Tp}$, its MSE is upper-bounded by that $\phi_{K}^{\min}$ (since it uses all snapshots, including $\x_{t_{\min}}$), and is lower-bounded by the MSE that would be obtained if $\tau_{\tp}=u_{\Tp}$ for $\tp=1,\cdots,\Tp$, and this MSE is $\Tp^{-1}$ times the MSE of  $\phi_{K}^{\min}$. Additionally, as said before, we have $\vzeta^{H} \R \vzeta \simeq \Tp CRB_{G}^{\Tp}(\phi)$  where $CRB_{G}^{\Tp}(\phitrue)$ is the Gaussian CRB using $\Tp$ snapshots. Hence, one can bound the MSE of $\phi_{K}^{\Tp}$ as
\begin{equation}
CRB_{G}^{\Tp}(\phitrue) C(\nu,\beta) \Tp^{-1/ \nu} \leq \E{\left( \phi_{K}^{\Tp} - \phitrue\right)^{2}}  \leq CRB_{G}^{\Tp}(\phitrue)C(\nu,\beta)\Tp^{-1/\nu +1}.
\end{equation}
As will be illustrated in the next section, the upper bound is rather tight, while the lower bound is much lower than the actual MSE.

\subsubsection{Estimation of $\R$ using secondary data}
When $\R$ is not known, then the secondary data $\Xs$ can be used to estimate it.  The maximum likelihood estimator is obtained (for $T \geq M$) as the solution (up to a scaling factor) to the following implicit equation \cite{Ollila12}
\begin{align}
\Rml &= \frac{1}{\Ts} \sum_{\ts=1}^{\Ts} \phi \left( \xst^{H} \Rml^{-1} \xst\right)  \xst \xst^{H} \nonumber \\
&=  \frac{1}{\beta^{1/2} \Ts} \sum_{\ts=1}^{\Ts}  \left( \xst^{H} \Rml^{-1} \xst\right)^{-1/2} \frac{K_{M+1-\nu}\left(2 \sqrt{\left( \xst^{H} \Rml^{-1} \xst\right) / \beta}\right)}{K_{M-\nu}\left(2 \sqrt{\left( \xst^{H} \Rml^{-1} \xst\right) / \beta}\right)} \xst \xst^{H}.
\end{align}
$\Rml$ can be obtained through an iterative procedure, whose convergence is guaranteed under the assumptions made \cite{Ollila12}. In order to avoid evaluation of the modified Bessel function, one can use the large $M-\nu$ approximation of $K_{M+1-\nu}\left(z\right) / K_{M-\nu}\left(z\right)$ in \eqref{approx_ratioK} to define $\Raml$ as the solution to the fixed-point solution 
\begin{equation}
\Raml = \frac{(M+1-\nu)^{1/2}(M-\nu)^{1/2}}{\Ts} \sum_{\ts=1}^{\Ts} \frac{\xst \xst^{H}}{\xst^{H} \Raml^{-1} \xst}. 
\end{equation}
Note that $\Raml$ is more or less the well-known Tyler fixed-point estimator \cite{Tyler87}, which again can be obtained from an iterative procedure whose convergence is guaranteed \cite{Pascal08,Chitour08}. The drawbacks of the two above estimators are that 1)they are suited to a $K$ distribution for the noise and 2)$\Ts$ is required to be larger than $M$. In order to gain robustness against these problems, a solution is to use normalized data $\zst = \xst / \left\| \xst \right\|$ whose distribution is independent of that of the noise, and to use regularization. More precisely, we suggest to resort to the following scheme \cite{Abramovich07c,Chen11,Wiesel12}
\begin{subequations}\label{FP_R(eta)}
\begin{align}
\breve{\R}_{k+1}(\eta) &=   (1-\eta) \frac{M}{\Ts} \sum_{\ts=1}^{\Ts} \frac{\zst \zst^{H}}{\zst^{H} \left[\Rhat_{k}(\eta) \right]^{-1} \zst}  + \eta \I_{M} \\
\Rhat_{k+1}(\eta) &= \frac{M}{\Tr{\breve{\R}_{k+1}(\eta)}}\breve{\R}_{k+1}(\eta).
\end{align}
\end{subequations}
and define $\Rracg (\eta) = \lim_{k \rightarrow \infty} \Rhat_{k}(\eta)$ since convergence of this iterative scheme has been proved \cite{Pascal14}. The very good performance of this scheme  has been illustrated in various applications, see e.g., \cite{Chen11,Wiesel12,Pascal14,Abramovich13,Besson13b}, where discussions on how to select the regularization parameter $\eta$ can also be found.

\section{Numerical simulations\label{section:numerical}}
We assume a linear array of $M=16$ elements spaced a half-wavelength apart and we consider the simple scenario of a single source impinging from $\phitrue=10^{\circ}$ embedded in unit power $K$-distributed noise. The covariance matrix $\R$ is given by $\R(k,\ell)=\rho^{|k-\ell|}$ with $\rho=0.99$. The exact and approximate maximum likelihood estimators, which consists in maximizing $ f(\phi)$ in \eqref{f(phi)} $\tilde{f}(\phi)$ in \eqref{ftilde(phi)} were implemented using the Matlab function \texttt{fminbnd}, and the maximum was searched in the interval $\left[ \phitrue- 2 \phi_{3dB} , \phitrue + 2\phi_{3dB} \right]$ where $\phi_{3dB}$ is the half-power beamwidth of the array. The signal waveform was generated from  i.i.d.  Gaussian variables with power $P$ and the signal to noise ratio (SNR) is defined as $SNR = P (\atrue^{H} \R^{-1} \atrue)$. The asymptotic Gaussian CRB, multiplied by the scalar $\alpha_{1} / M$   was used as the bound for $K$-distributed noise.   For the regularized covariance matrix estimator $\Rracg (\eta)$ of \eqref{FP_R(eta)}, the value of $\eta$ was set to $\eta=0.01$. $1000$ Monte-Carlo simulations were used to evaluate the mean-square error (MSE) of the estimates.

In Figures \ref{fig:MSE_vs_T_Ts=32_nu>1} and \ref{fig:MSE_vs_T_Ts=32_nu<1} we plot the CRB (for $\nu > 1$) or the lower and upper bounds of \eqref{lowerbound_besselk} when $\nu < 1$,  as well as the MSE of the ML and AML estimators, as a function of $\Tp$, and compare the case where $\R$ is known to the case where it is estimated from \eqref{FP_R(eta)} with $\Ts=32$ snapshots in the secondary data. The following observations can be made:
\begin{itemize}
\item there is almost no difference between the MLE and the AMLE, and therefore the latter should be favored since it does not require evaluating modified Bessel functions.
\item the MSE in the case where $\R$ is known is lower than that when $\R$ is to be estimated, which is expected. However, the difference is smaller when $\nu < 1$: in other words, it seems that adaptive whitening is not so much penalizing with small $\nu$ while it seems more crucial for $\nu  > 1$. Indeed, for small $\nu$, what matters most is the fact that some snapshots are nearly noiseless, and this is more influential than obtaining a very good whitening.
\item the decrease of the MSE for $\nu > 1$ is roughly of the order $\Tp^{-1}$. When $\nu < 1$, this rate is significantly increased and the MSE decreases very quickly as $\Tp^{-1/ \nu}$, as predicted by the analysis above. This rate of convergence is also observed in Figure \ref{fig:MSE_multi_Rknown_vs_T_vs_nu} where we consider a scenario with two sources at $\phi=10^{\circ},12^{\circ}$.
\item the upper bound in \eqref{lowerbound_besselk} seems to provide quite a good approximation of the actual MSE, at least for $\Tp$ large enough.
\end{itemize}

The influence of $\Ts$ is investigated in Figure \ref{fig:MSE_vs_Ts_T=16}, where one can observe that about $\Ts=64$ is necessary for the performance with estimated $\R$ to be very close to the performance for known $\R$. However, as indicated above, this is less pronounced when $\nu < 1$, where the difference becomes smaller with lower $\Ts$.

Finally, we investigate whether the rate of convergence of the MLE or AMLE when $\nu$ varies is impacted by a small amount of Gaussian noise. More precisely, we run simulations where the data is generated as
\begin{equation}
\xpt = \va(\phitrue) s_{\tp} + \sqrt(1-\alpha) \sqrt{\tau_{\tp}} \R^{1/2} \wpt + \sqrt{\alpha} \vpt
\end{equation}
where $\wpt , \vpt \sim \vCN{\vect{0}}{\I_{M}}$, i.e., the noise is a mixture of $K$-distributed noise and Gaussian distributed noise. The covariance matrix of the noise is now $\R_{K+G} = (1-\alpha) \R + \alpha \I$ and we use the AML estimator assuming that the noise has a $K$ distribution with parameter $\nu$ and known covariance matrix $\R_{K+G}$. In Figure \ref{fig:MSE_mixture_vs_T}, we display the MSE of the AML estimator versus $\Tp$ and versus $\nu$ for different values of $\alpha$. Clearly, the rate of convergence of the estimator is affected  by a small amount of Gaussian noise, even when $\nu$ is small. This indicates that, if noise is not purely $K$-distributed with small $\nu$, we recover the usual behavior of the MSE versus $\Tp$.

\section{Conclusions \label{section:conclu}}
In this paper we addressed the DoA estimation problem in $K$-distributed noise using two data sets. The main result of the paper was to show that, when the shape parameter $\nu$ of the texture Gamma distribution is below $1$, the FIM is unbounded. On the other hand, for $\nu > 1$, the FIM is bounded and we derived an accurate closed-form approximation of the CRB. The maximum likelihood estimator was derived as well as an approximation, which induces non significant losses compared to the  exact MLE. In the non regular case where $\nu < 1$, we derived  lower and  upper  bounds on the mean-square error of the (A)ML estimates and we showed that the rate of convergence of these (A)ML estimates is about $\Tp^{-1/\nu}$ where $\Tp$ is the number of snapshots.

\begin{figure}[p]
 \centering
  \subfigure[][]{\label{fig:MSE_vs_T_Ts=32_nu=10_SNR=3}
  \includegraphics[width=7.5cm]{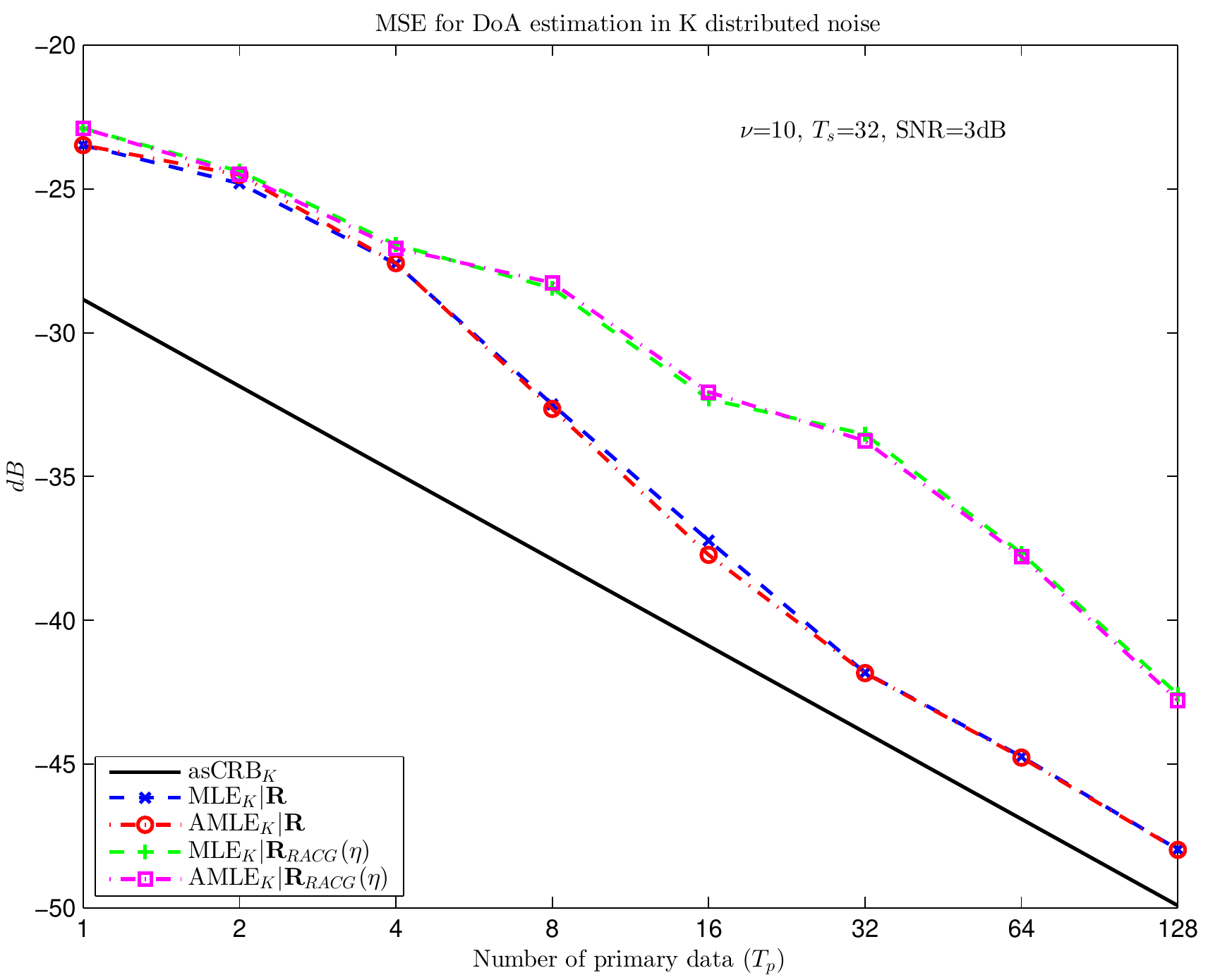}}
   \subfigure[][]{\label{fig:MSE_vs_T_Ts=32_nu=2_SNR=3}
  \includegraphics[width=7.5cm]{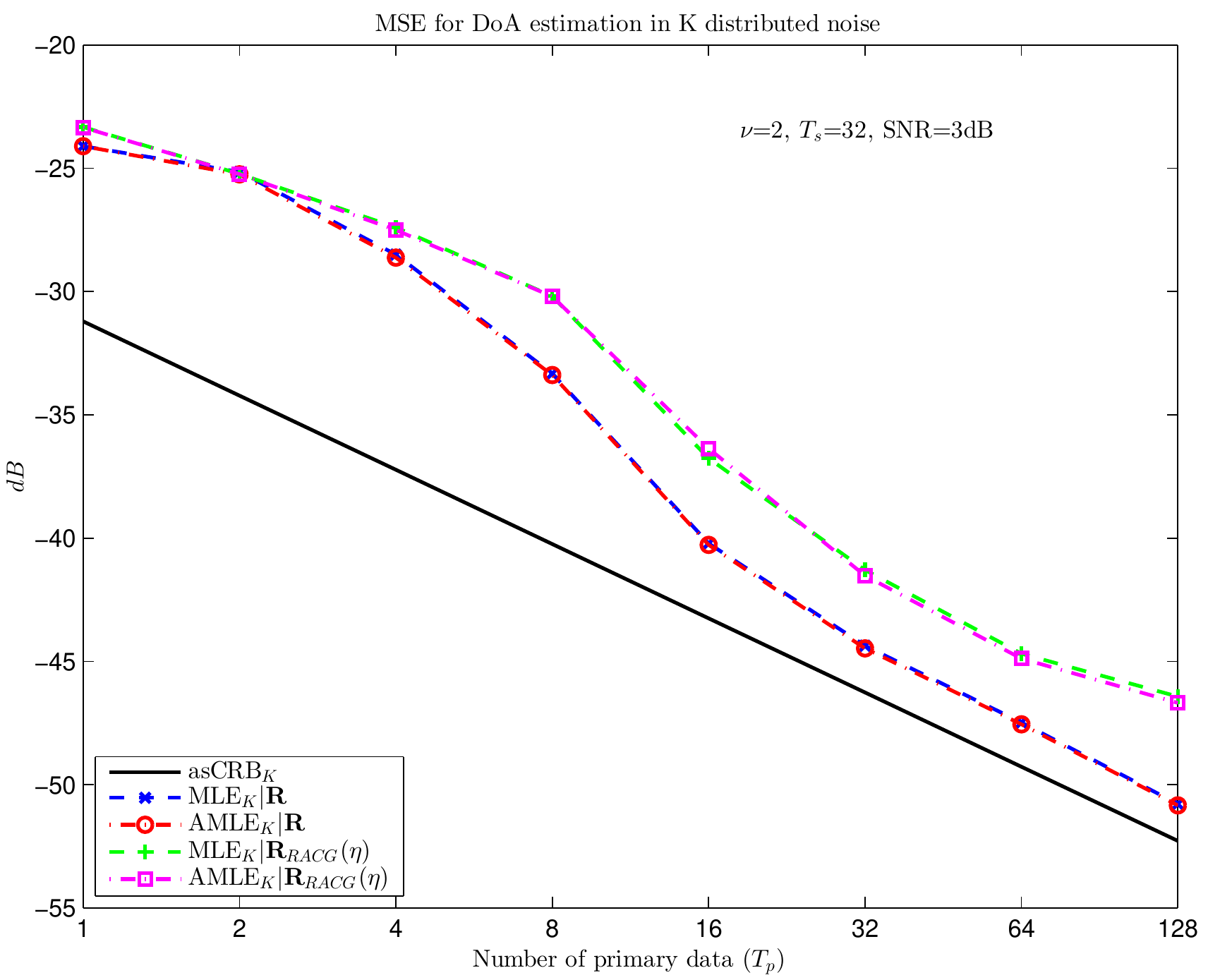}} \\  
    \subfigure[][]{\label{fig:MSE_vs_T_Ts=32_nu=1.5_SNR=3}
  \includegraphics[width=7.5cm]{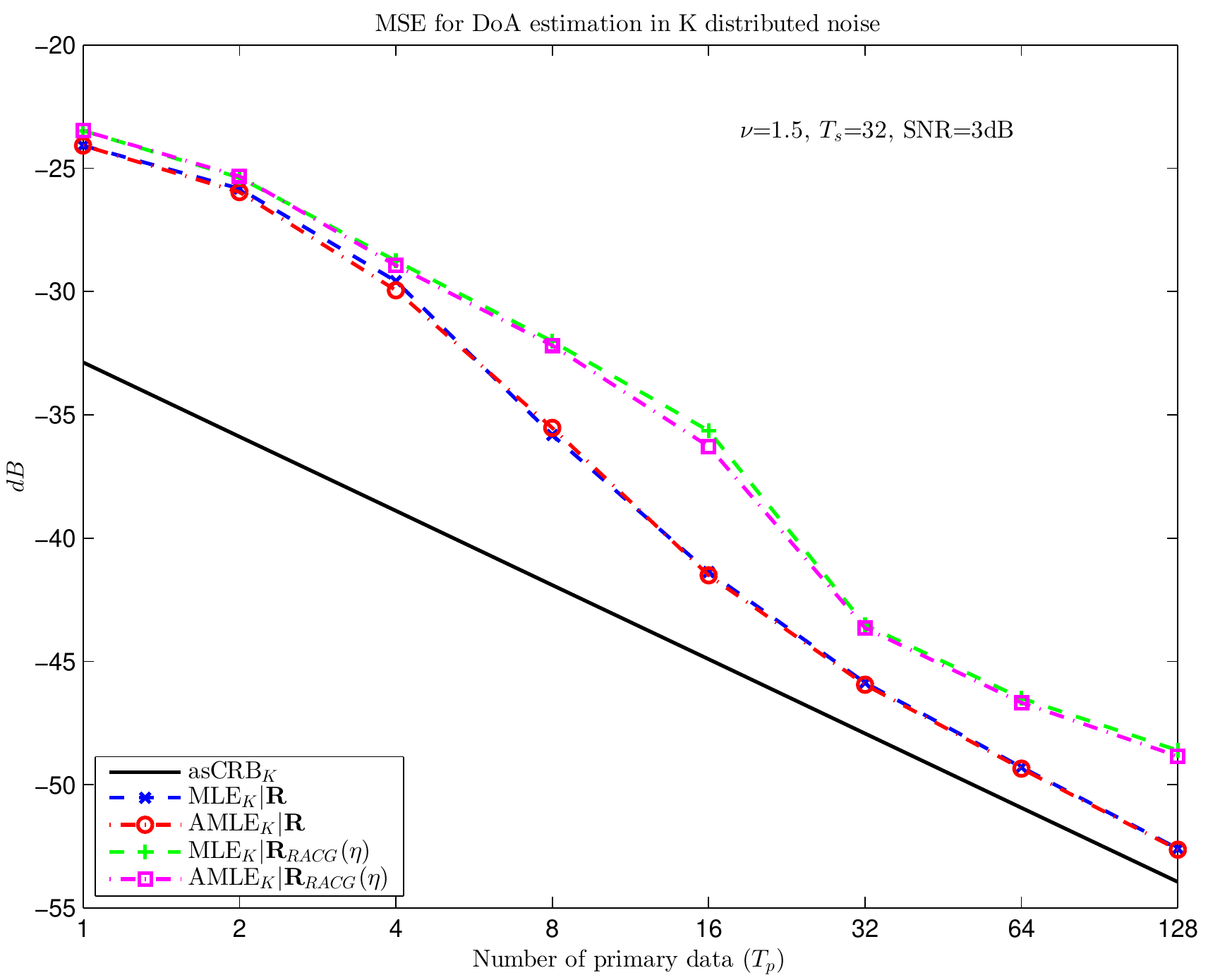}}
   \subfigure[][]{\label{fig:MSE_vs_T_Ts=32_nu=1.1_SNR=3}
  \includegraphics[width=7.5cm]{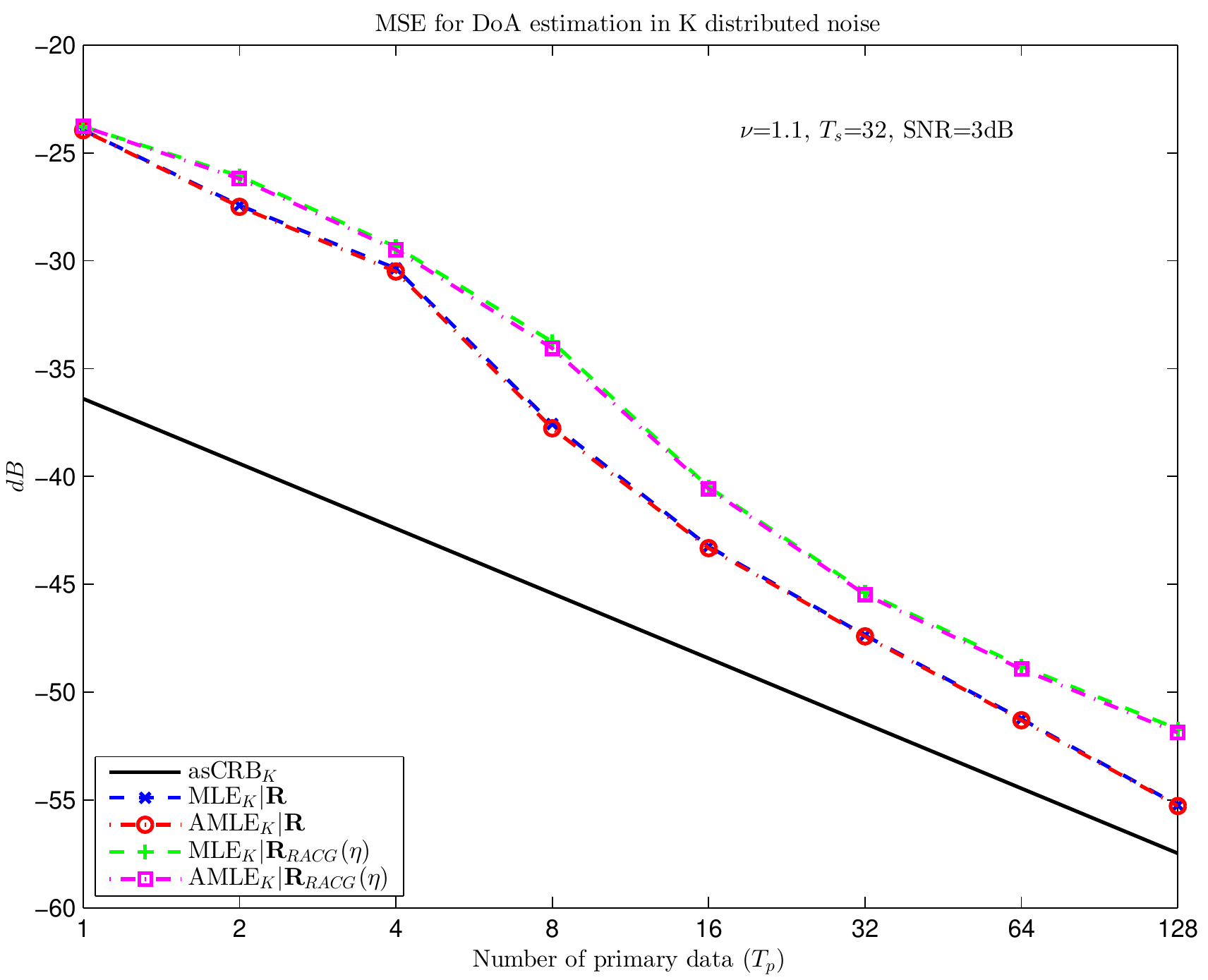}} \\ 
  \caption{Cram\'{e}r-Rao bounds and mean square error of  estimators versus $\Tp$ with either $\R$ known or estimated. $M=16$, $SNR=3$dB,  $\Ts=32$ and $\nu > 1$.}
  \label{fig:MSE_vs_T_Ts=32_nu>1}
\end{figure}

\begin{figure}[p]
 \centering
   \subfigure[][]{\label{fig:MSE_vs_T_Ts=32_nu=0.9_SNR=3}
  \includegraphics[width=7.5cm]{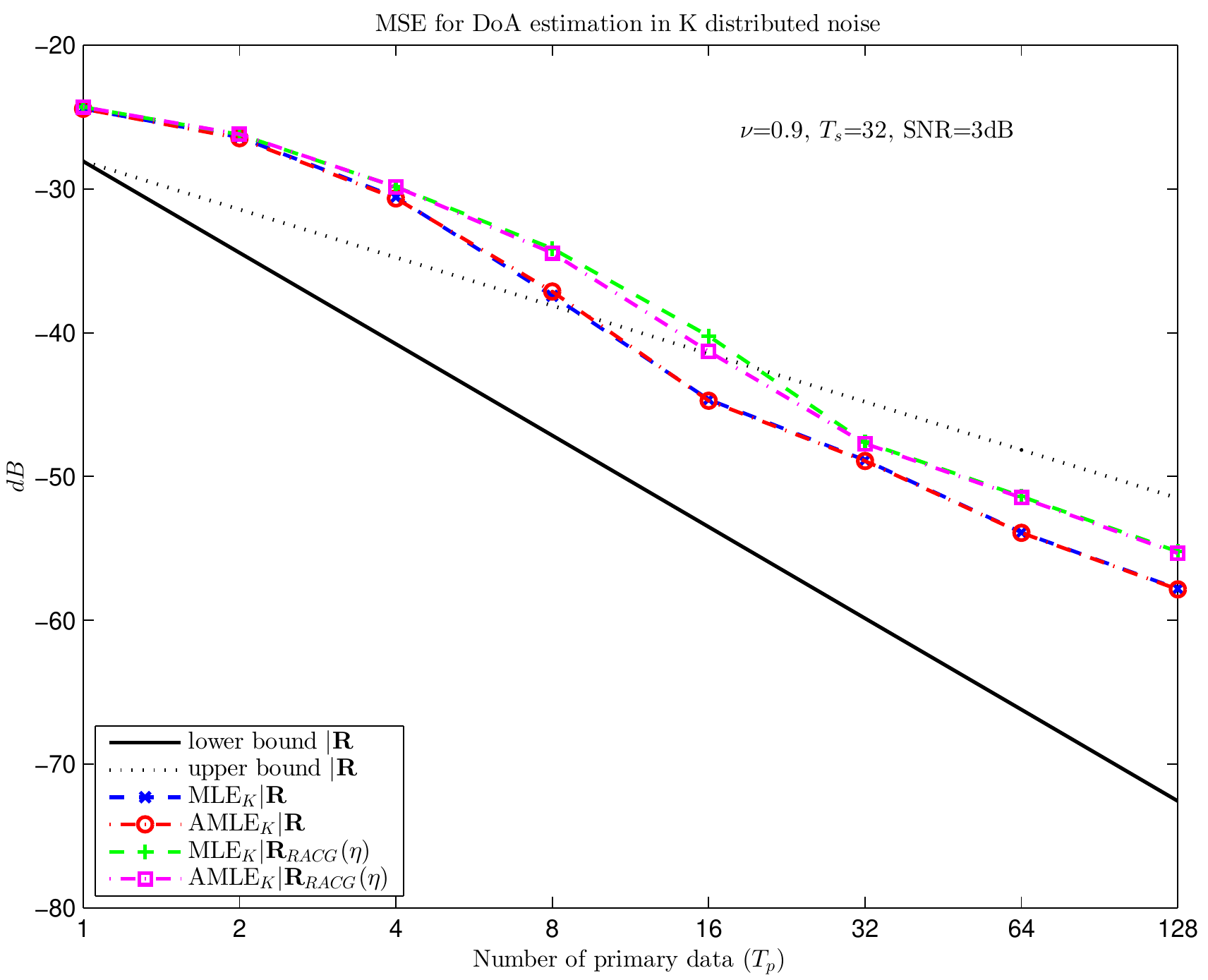}}
   \subfigure[][]{\label{fig:MSE_vs_T_Ts=32_nu=0.5_SNR=3}
  \includegraphics[width=7.5cm]{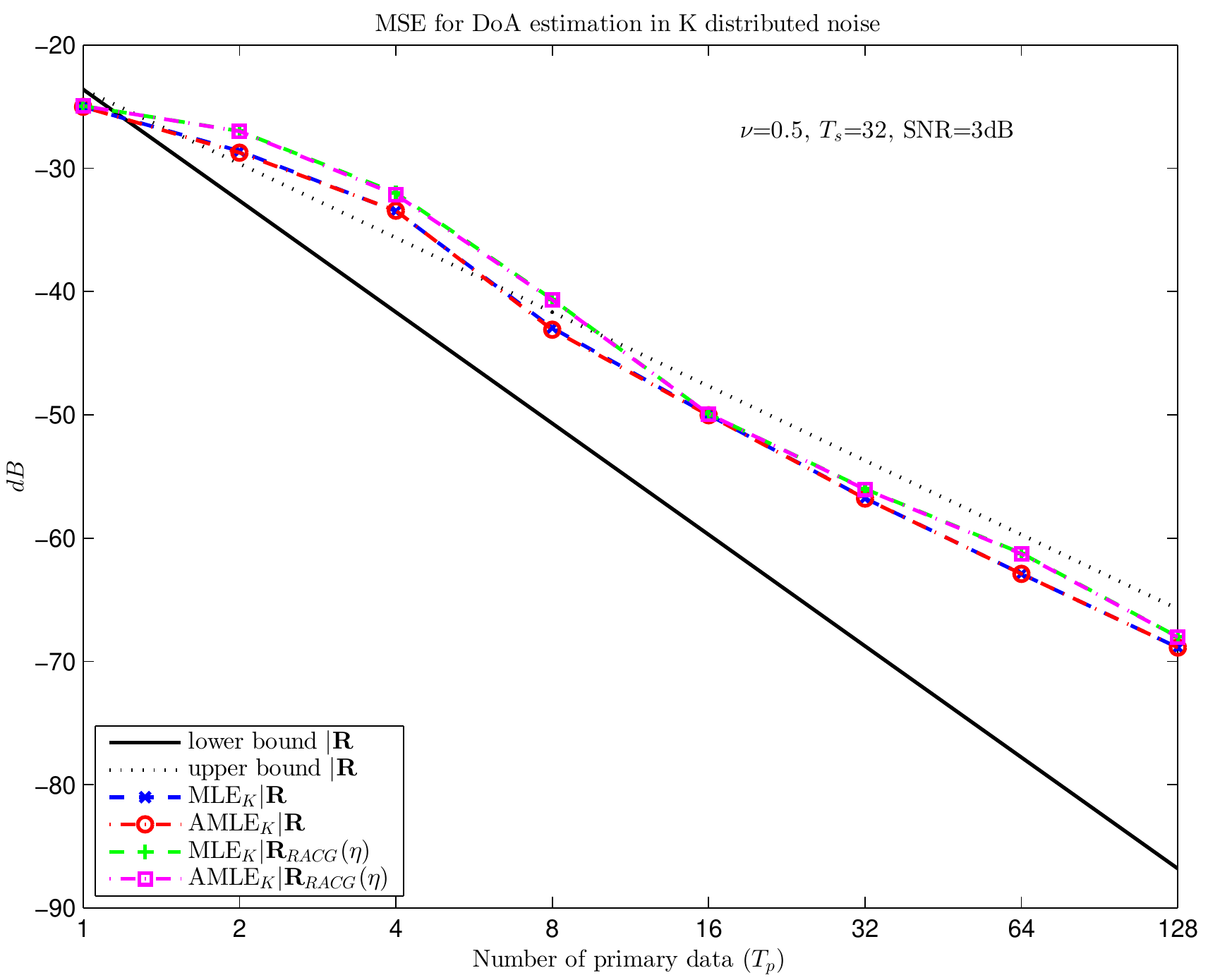}} \\
    \subfigure[][]{\label{fig:MSE_vs_T_Ts=32_nu=0.2_SNR=3}
  \includegraphics[width=7.5cm]{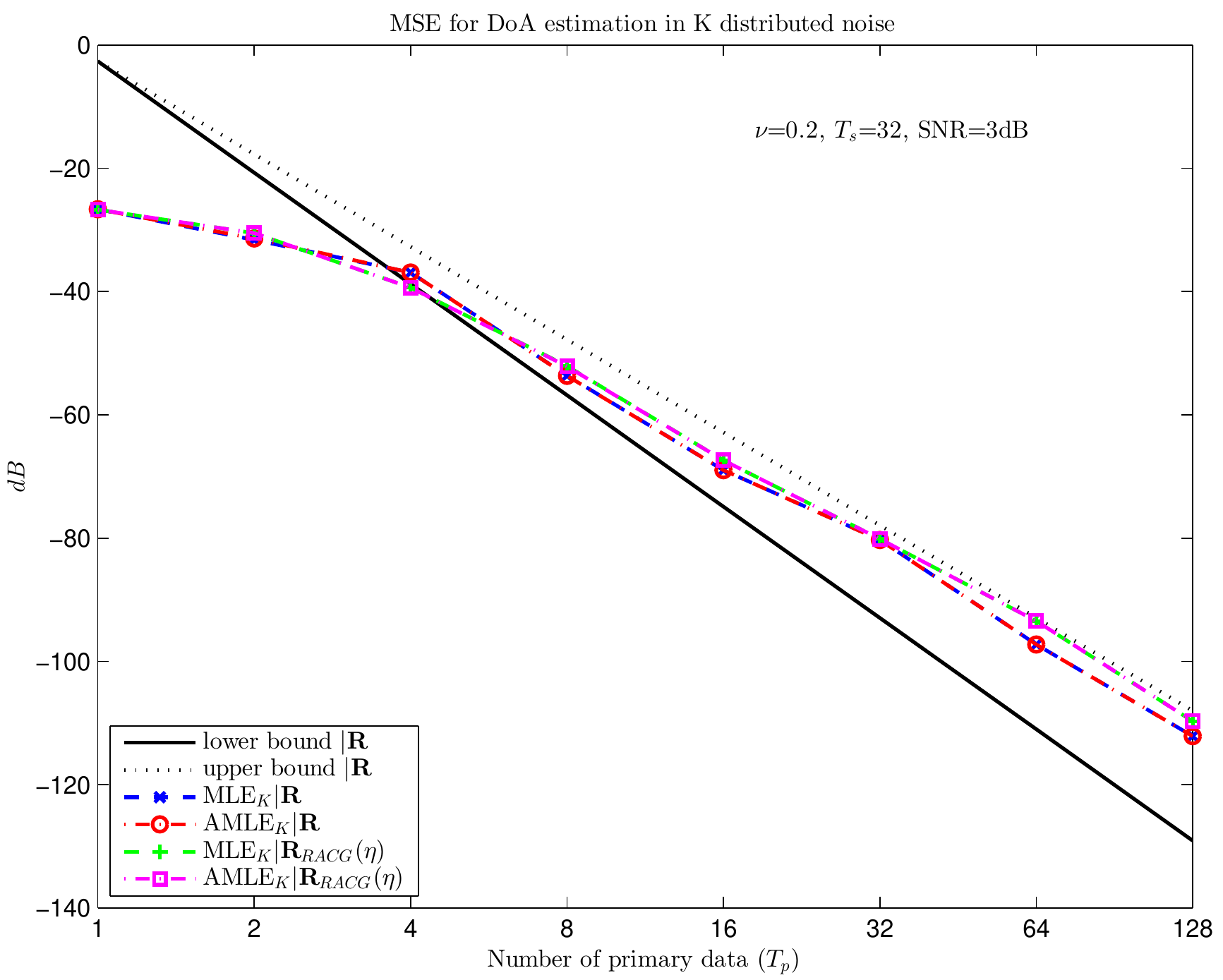}}
   \subfigure[][]{\label{fig:MSEvs_T_Ts=32_nu=0.1_SNR=3}
  \includegraphics[width=7.5cm]{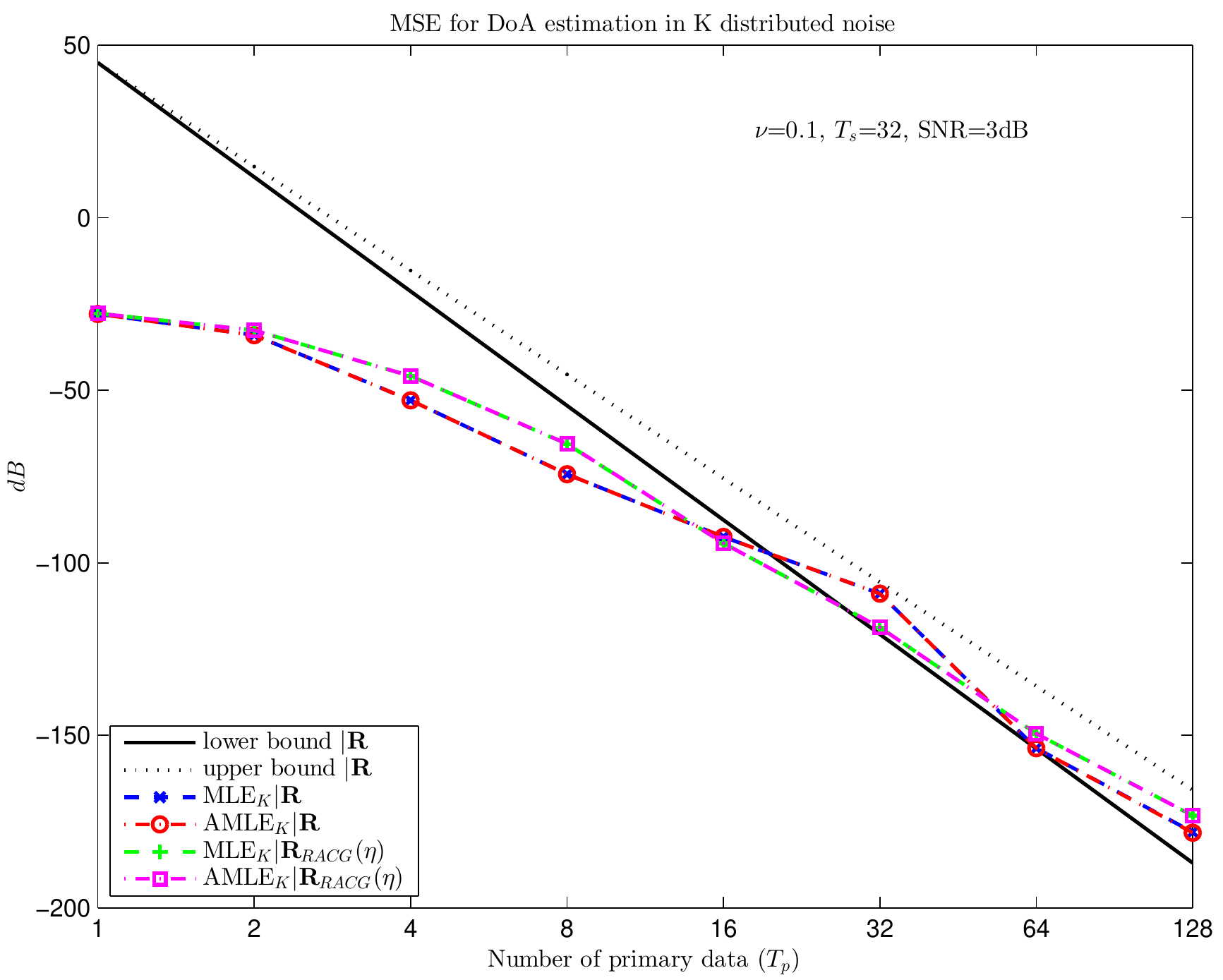}} \\  
  \caption{Mean square error of  estimators versus $\Tp$ with either $\R$ known or estimated. $M=16$, $SNR=3$dB, $\Ts=32$ and $\nu \leq 1$.}
   \label{fig:MSE_vs_T_Ts=32_nu<1}
\end{figure}

\begin{figure}[htb]
\centering
\includegraphics[width=8cm]{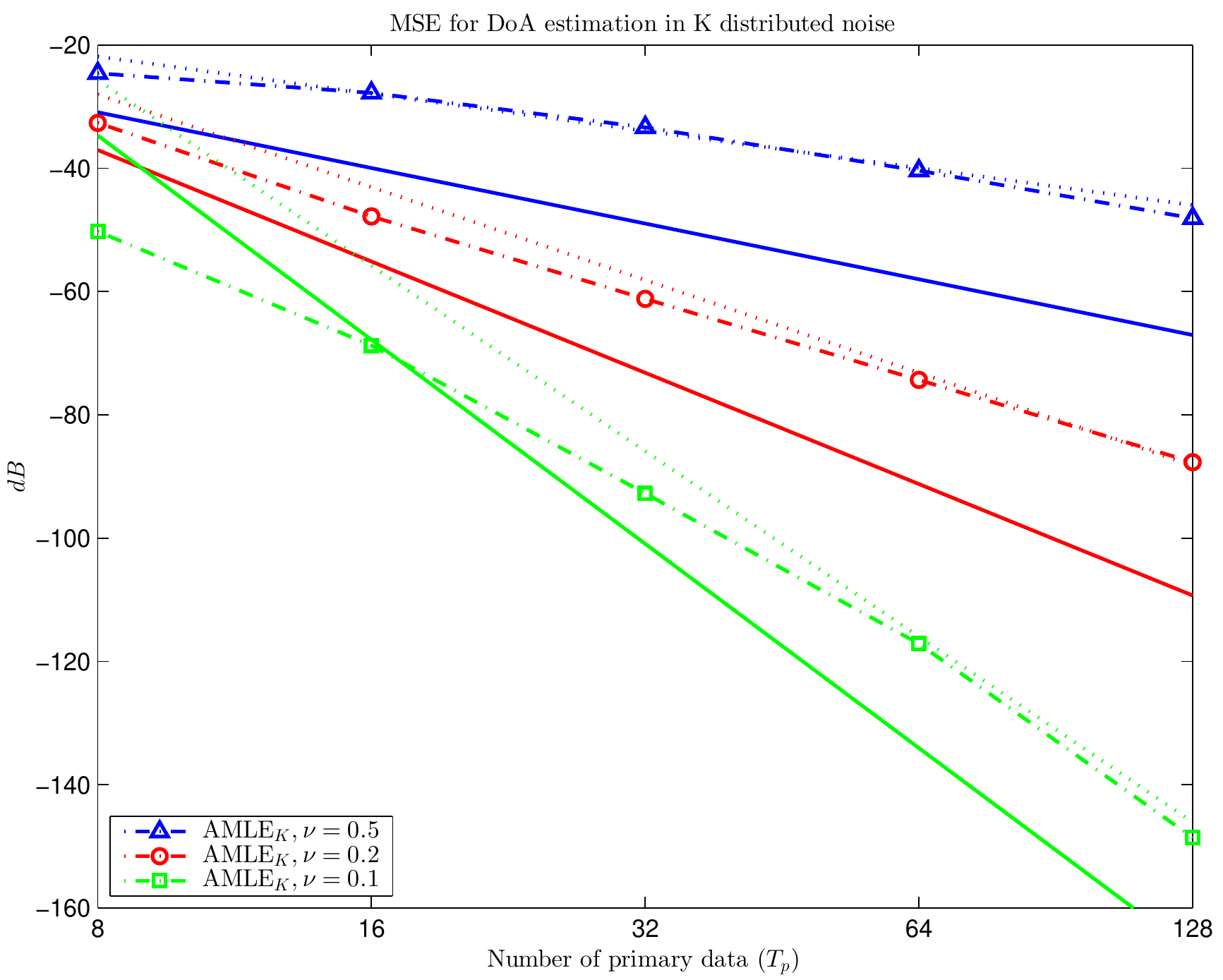}
\caption{Mean square error of  AML estimator in a two sources scenario with $\phi=10^{\circ},12^{\circ}$. $\R$ known,  $M=16$ and $SNR=3$dB.}
\label{fig:MSE_multi_Rknown_vs_T_vs_nu}
\end{figure}

\begin{figure}[p]
 \centering
  \subfigure[][]{\label{fig:MSE_vs_Ts_T=16_nu=10_SNR=3}
  \includegraphics[width=7.5cm]{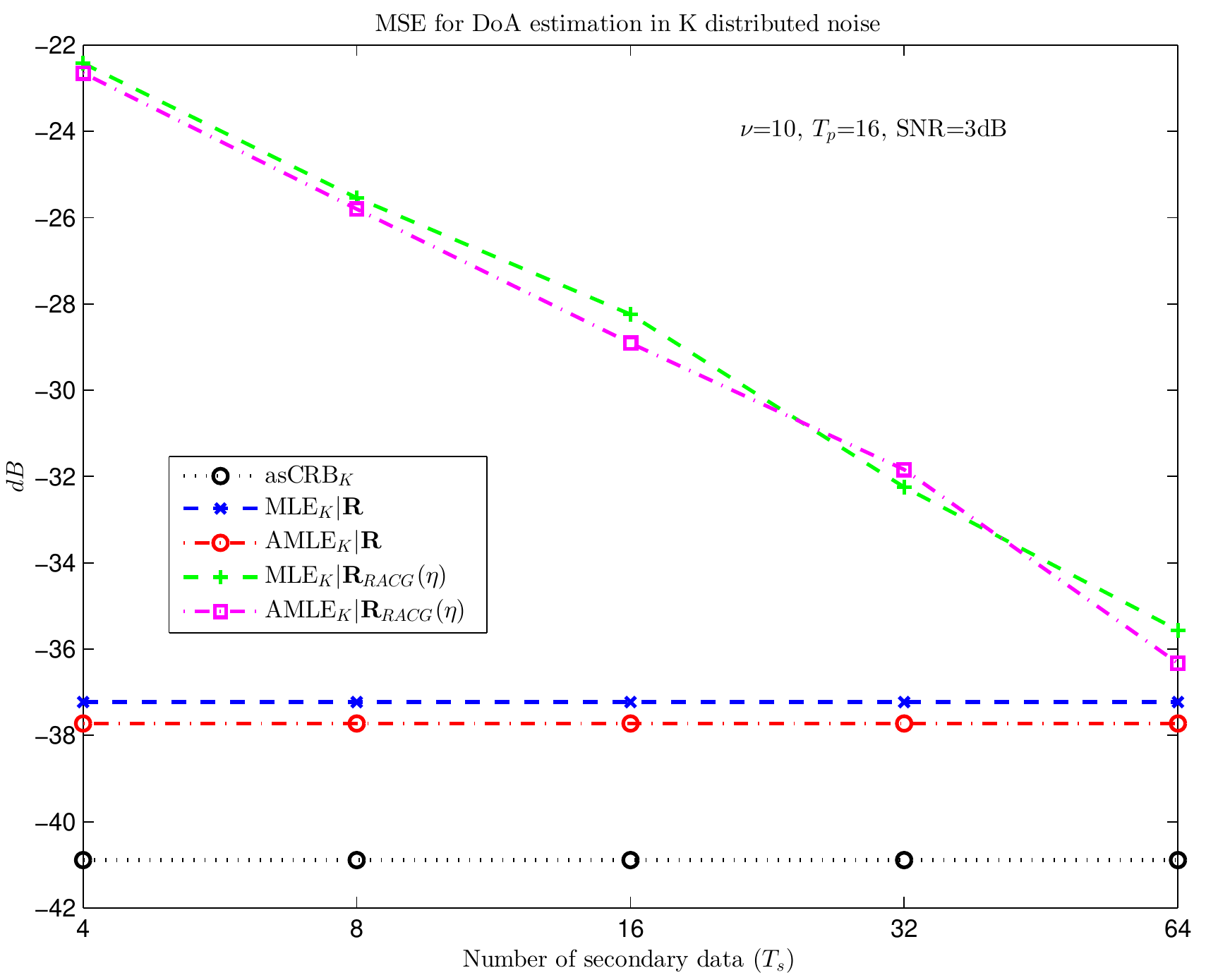}}
   \subfigure[][]{\label{fig:MSE_vs_Ts_T=16_nu=1.5_SNR=3}
  \includegraphics[width=7.5cm]{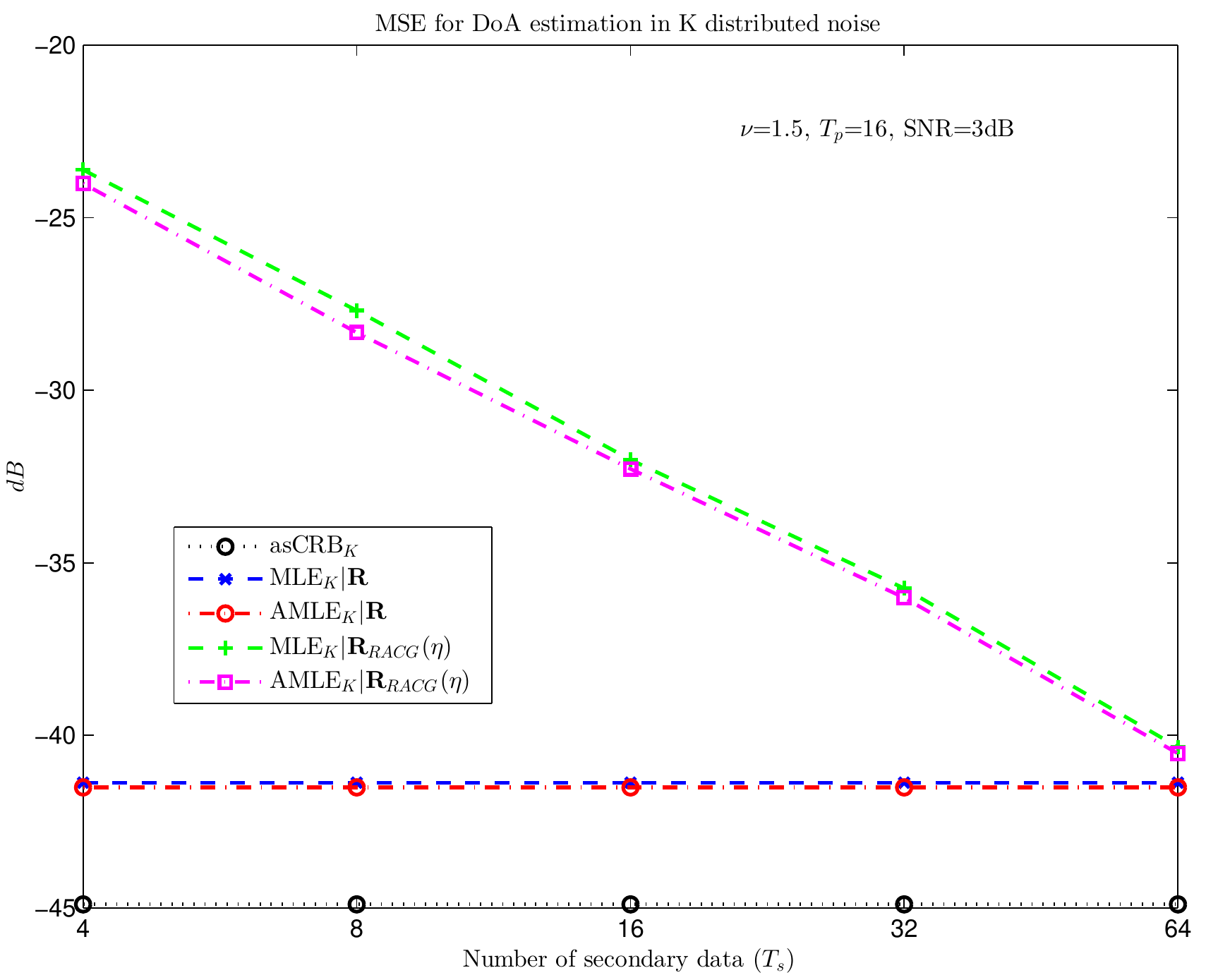}} \\  
    \subfigure[][]{\label{fig:MSE_vs_Ts_T=16_nu=0.9_SNR=3}
  \includegraphics[width=7.5cm]{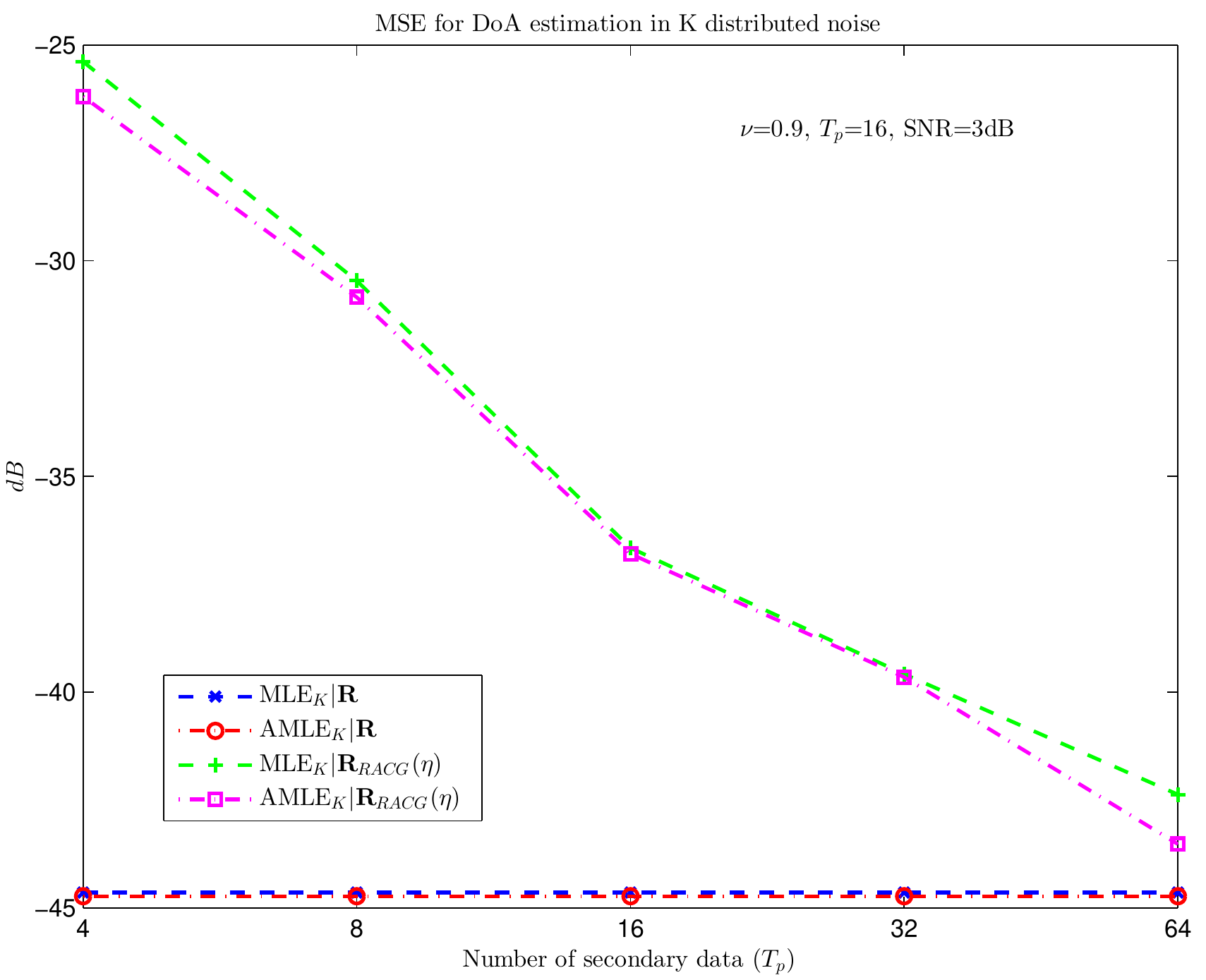}}
   \subfigure[][]{\label{fig:MSE_vs_Ts_T=16_nu=0.5_SNR=3}
  \includegraphics[width=7.5cm]{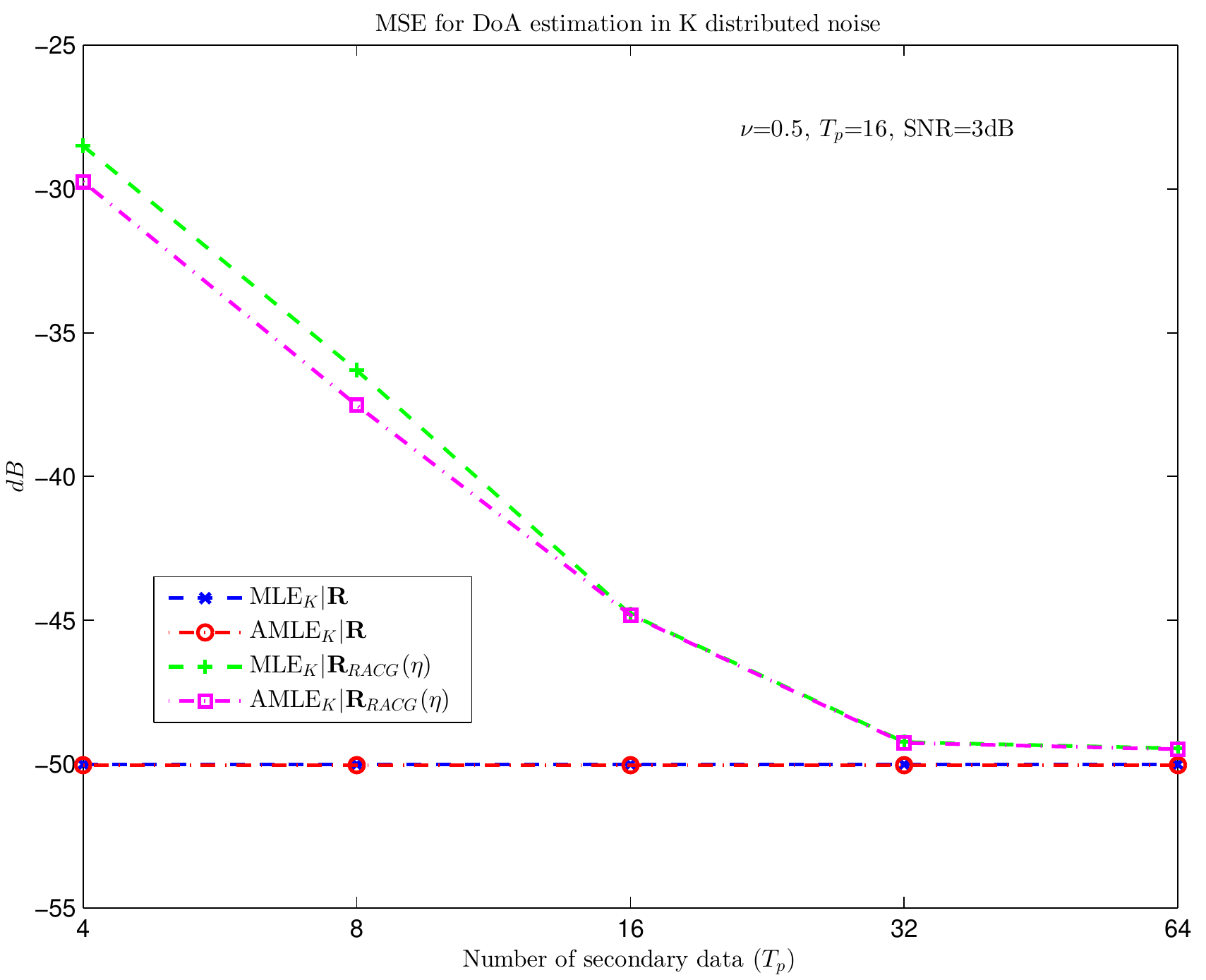}} \\ 
  \caption{Cram\'{e}r-Rao bounds and mean square error of  estimators versus $\Ts$. $M=16$, $SNR=3$dB, $\Tp=16$ and varying $\nu$.}
  \label{fig:MSE_vs_Ts_T=16}
\end{figure}

\begin{figure}[p]
 \centering
  \subfigure[][]{\label{fig:MSE_mixture_vs_T_alpha=0_SNR=3}
  \includegraphics[width=7.5cm]{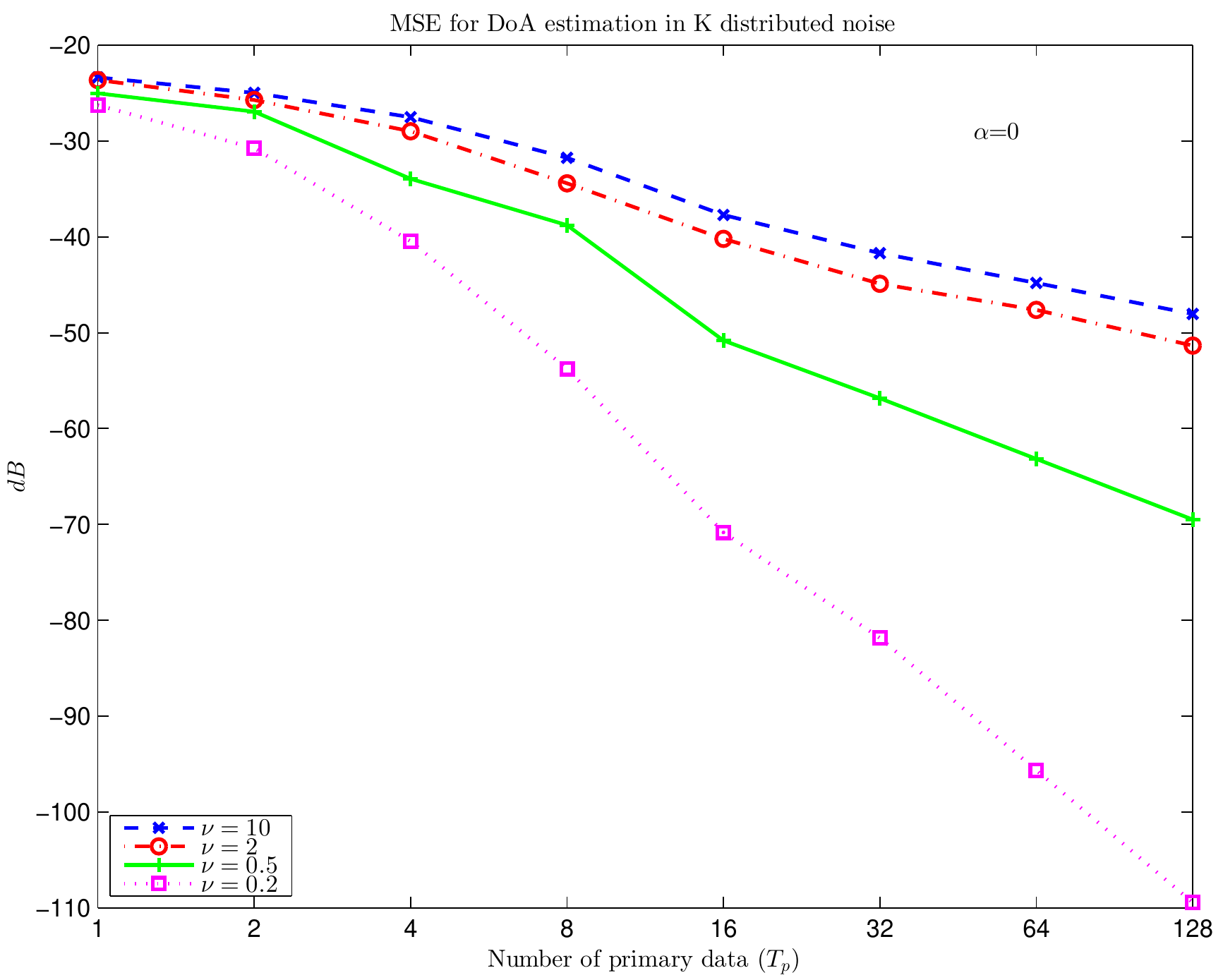}}
   \subfigure[][]{\label{fig:MSE_mixture_vs_T_alpha=0.01_SNR=3}
  \includegraphics[width=7.5cm]{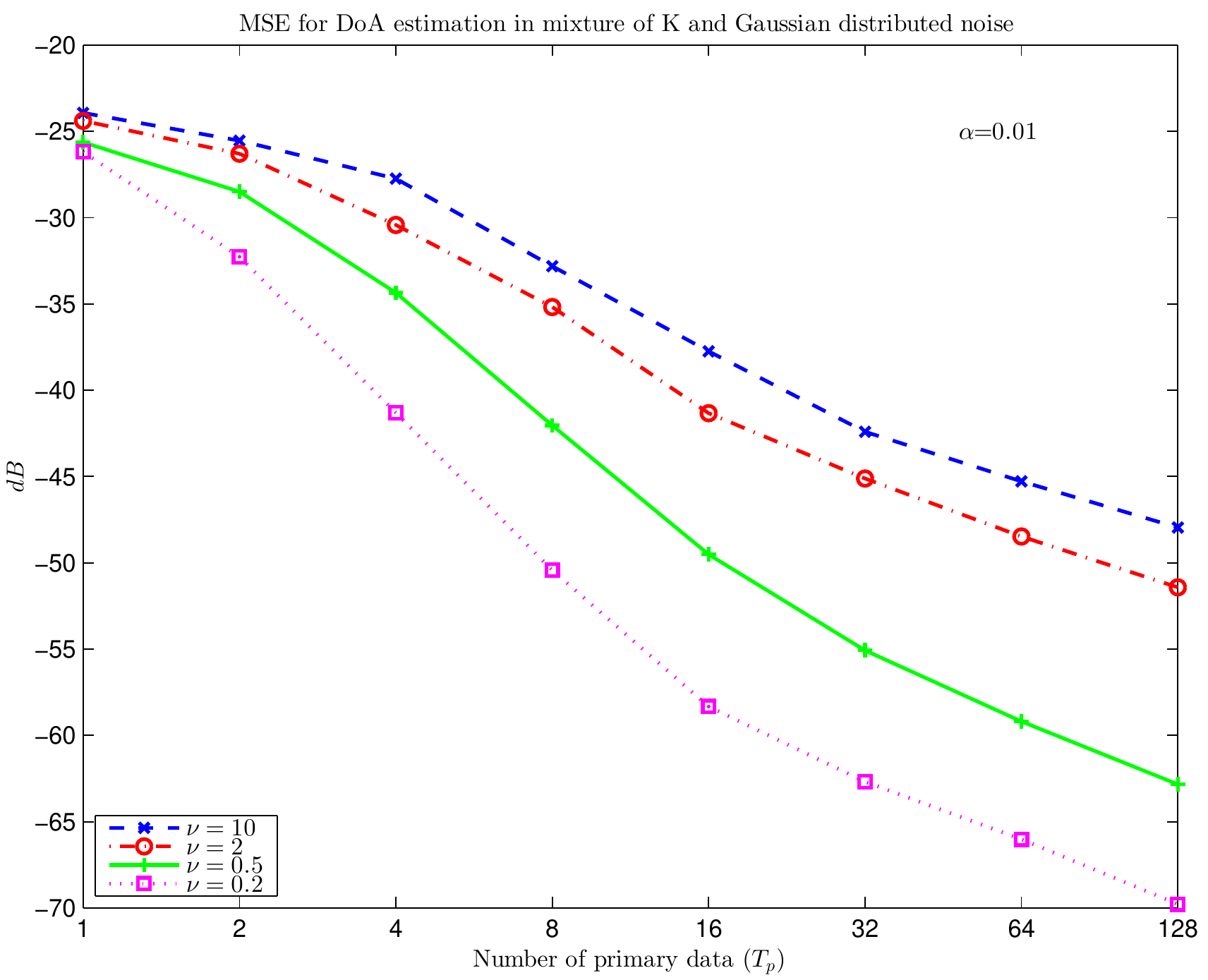}} \\  
    \subfigure[][]{\label{fig:MSE_mixture_vs_T_alpha=0.1_SNR=3}
  \includegraphics[width=7.5cm]{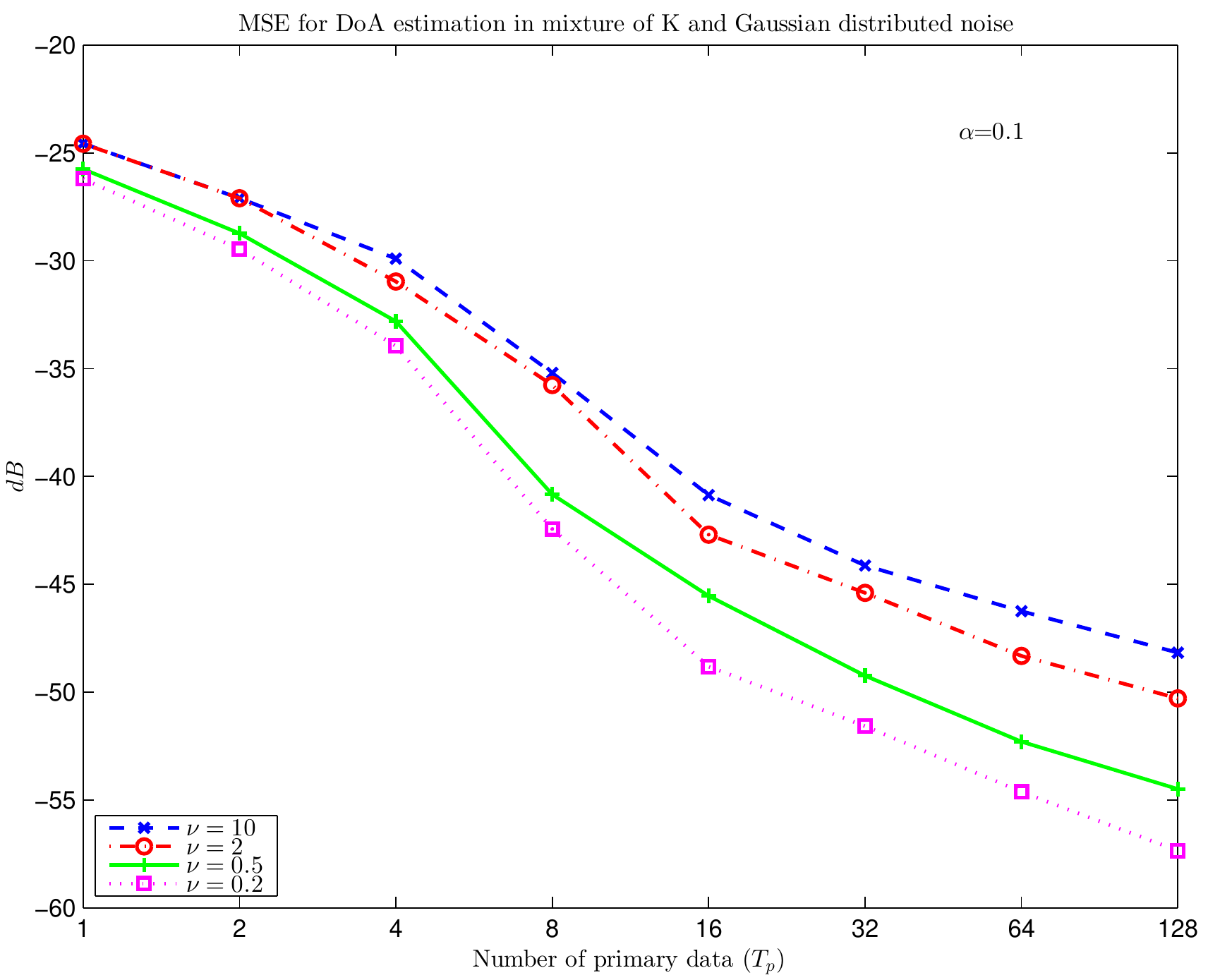}}
   \subfigure[][]{\label{fig:MSE_mixture_vs_T_alpha=0.2_SNR=3}
  \includegraphics[width=7.5cm]{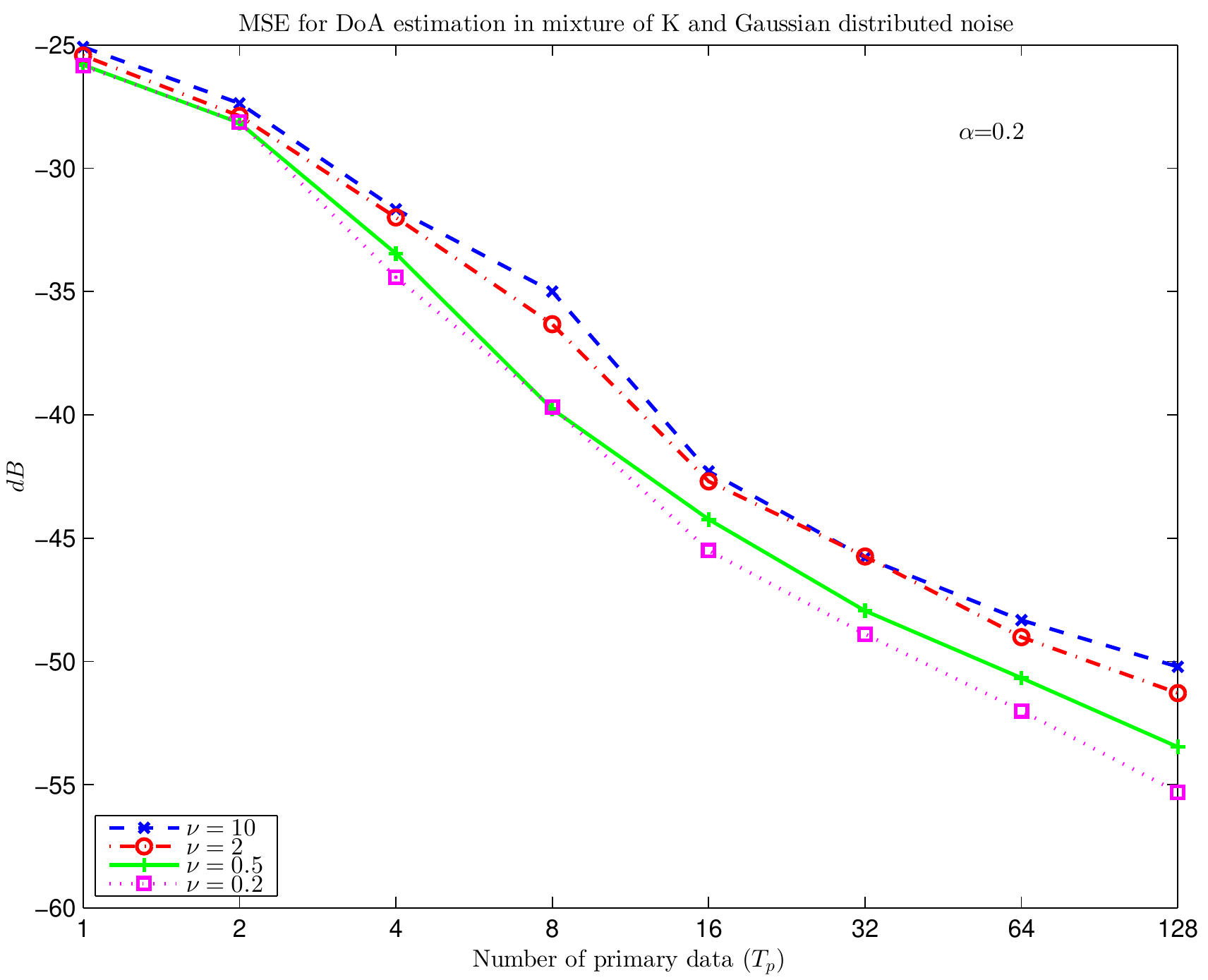}} \\ 
  \caption{Mean square error of AMLE versus $\Tp$ in the case of a mixture of $K$-distributed and Gaussian distributed noise. $M=16$, $SNR=3$dB,  and varying $\alpha$.}
  \label{fig:MSE_mixture_vs_T}
\end{figure}

\end{document}